\documentclass[10pt,twocolumn,aps,english,pra,superscriptaddress,longbibliography]{revtex4-2} 

\usepackage[utf8]{inputenc}

\usepackage{tikz} % Package for drawing
\usepackage{amsmath,amsfonts,amsthm, mathtools}
\usepackage{placeins}
\usepackage{float}
\usepackage{dsfont}
\usepackage{csquotes}
\usepackage{amssymb}
\usepackage{verbatim}
\usepackage{bm}
\usepackage{amsmath}
\usepackage{amssymb}
\usepackage{stmaryrd}
\usepackage{amsthm}
\usepackage[caption=false]{subfig}
\usepackage{physics}
\usepackage{bm}  

\usepackage{hyperref}
\definecolor{darkred}  {rgb}{0.5,0,0}
\definecolor{darkblue} {rgb}{0,0,0.5}
\definecolor{darkgreen}{rgb}{0,0.5,0}
\hypersetup{
  colorlinks = true,
  urlcolor  = blue,         % color of external links
  linkcolor = darkblue,     % color of internal links
  citecolor = darkgreen,    % color of links to bibliography
  filecolor = darkred       % color of file links
}

\theoremstyle{definition}

\newtheorem{lemma}{Lemma}
\newtheorem{proposition}{Proposition}
\newtheorem{thm}{Theorem}

\usepackage{multirow}

\definecolor{cool_green}{rgb}{0.0, 0.5, 0.0}

\begin{document}

%\setlength{\parindent}{0pt}
%\title{Quantifying performance of quantum imaging using learning theory}
\title{Secure quantum ranging}
%\date{}
\author{Yunkai Wang}
\email{ywang10@perimeterinstitute.ca}
\affiliation{Perimeter Institute for Theoretical Physics, Waterloo, Ontario N2L 2Y5, Canada.}
\affiliation{Institute for Quantum Computing, University of Waterloo, Ontario N2L 3G1, Canada.}
\affiliation{Department of Applied Mathematics, University of Waterloo, Ontario N2L 3G1, Canada.}

\author{Graeme Smith}
\email{graeme.smith@uwaterloo.ca}
\affiliation{Institute for Quantum Computing, University of Waterloo, Ontario N2L 3G1, Canada.}
\affiliation{Department of Applied Mathematics, University of Waterloo, Ontario N2L 3G1, Canada.}

\author{Alex May}
\email{amay@perimeterinstitute.ca}
\affiliation{Perimeter Institute for Theoretical Physics, Waterloo, Ontario N2L 2Y5, Canada.}
\affiliation{Institute for Quantum Computing, University of Waterloo, Ontario N2L 3G1, Canada.}
%\affiliation{Department of Physics and Astronomy, University of Waterloo, Ontario N2L 3G1, Canada.}

\begin{abstract}
Determining and verifying an object's position is a fundamental task with broad practical relevance. We propose a secure quantum ranging protocol that combines quantum ranging with quantum position verification (QPV). Our method achieves Heisenberg-limited precision in position estimation while simultaneously detecting potential cheaters. Two verifiers each send out a state that is entangled in frequency space within a single optical mode.
An honest prover only needs to perform  simple beam-splitter operations, whereas cheaters are allowed to use arbitrary linear optical operations, one ancillary mode, and perfect quantum memories—though without access to entanglement. Our approach considers a previously unstudied security aspect to quantum ranging. It also provides a framework to quantify the precision with which a prover's position can be verified in QPV, which previously has been assumed to be infinite.  
\end{abstract}

\maketitle

\textit{Introduction} - Ranging determines the distance to a target using probe states that acquire distance-dependent information through interactions or operations at the target. In recent years, there has been growing interest in exploring how quantum technologies can enhance this task. In particular, much attention has been given to how entanglement can improve the precision of position estimation \cite{giovannetti2001quantum,maccone2020quantum}. When entangled probe states are sent toward an object with an unknown position and the reflected states are measured, the ranging precision can achieve Heisenberg-limited scaling with respect to the number of photons used. Beyond improving precision, recent studies have investigated whether entangled states can enable additional functionalities in ranging. For example, it has been proposed that one can simultaneously detect the presence of an object and estimate its distance \cite{zhuang2022ultimate,zhuang2021quantum}, bridging the tasks of quantum ranging and quantum illumination \cite{tan2008quantum,lloyd2008enhanced,barzanjeh2015microwave,nair2020fundamental,sanz2017quantum,barzanjeh2020microwave,gregory2020imaging,karsa2024quantum,shapiro2020quantum}. We further extend this line of research by integrating a new feature into quantum ranging: the verification of an object's position to ensure security against potential spoofing or cheating attempts.

The security of verifying a prover's position has been extensively studied under the framework of QPV \cite{beausoleil2006tagging, malaney2010quantum, malaney2010location, kent2011quantum, buhrman2014position, vaidman2003instantaneous, beigi2011simplified, bluhm2022single, asadi2025linear,gonzales2019bounds, speelman2016instantaneous, chakraborty2015practical}. In a typical QPV scenario, two verifiers— referred to as Alice and Bob—seek to confirm the location of a prover, Charlie. They send quantum states at the speed of light, along with instructions that require Charlie to perform certain operations and return the results immediately. Because these operations must be performed instantaneously, any cheaters located elsewhere would need entanglement resources to simulate Charlie’s responses. This concept was first introduced in a patent \cite{beausoleil2006tagging} and later developed in the academic literature \cite{malaney2010quantum,malaney2010location,kent2011quantum,buhrman2014position}. Subsequent work has shown that, with sufficiently complex states and operations, the entanglement required to successfully cheat becomes substantial \cite{vaidman2003instantaneous,beigi2011simplified,bluhm2022single,asadi2025linear,gonzales2019bounds,speelman2016instantaneous,chakraborty2015practical,tomamichel2013monogamy}. When cheaters are limited in entanglement resources—as is typically the case in practice—secure position verification becomes feasible. However, research in QPV has largely focused on the quantum computing aspects, particularly on designing attacks using entanglement, while implicitly assuming that the prover’s position can be verified with infinite precision.

In this work, we bridge quantum ranging and QPV, introducing a protocol that not only estimates the position of the prover Charlie with precision achieving Heisenberg scaling over photon number, but also ensures security against cheaters attempting to spoof the position of Charlie—effectively realizing a secure form of quantum ranging. At the same time, our approach introduces a metrological framework within QPV, enabling explicit analysis of the achievable precision in position verification—an aspect that has received little attention despite its practical relevance. Furthermore, unlike some conventional QPV schemes that rely on multi-qubit states to enhance security,  we demonstrate that entanglement within a single optical mode per side can already strengthen security. Our results thus offer a new perspective on both the limitations and potential improvements of QPV, particularly in optical implementations.

\textit{Set up and ranging precision} - We consider a scenario in which two verifiers, Alice and Bob, aim to estimate the position of a prover, Charlie. As illustrated in Fig.~\ref{set_up}(a), Alice and Bob prepare a bipartite state randomly drawn from an ensemble $\{p_i, \ket{\psi_i}\}$ and send it to Charlie. The potential cheaters are assumed to know the ensemble used by Alice and Bob, but not the specific state chosen in each instance.  Upon receiving the state, Charlie applies an instantaneous operation $U$, which is predetermined and known to all parties. Since this operation acts on quantum states with specific pulse shapes, it implicitly depends on Charlie’s position $y$. We assume that Alice and Bob are located at positions $\pm L/2$. After applying the operation, Charlie returns the transformed state $\ket{\phi_i} = U \ket{\psi_i}$ to Alice and Bob. They then perform measurements on the received state $\ket{\phi_i}$ to estimate Charlie’s position.

To enhance the precision of estimating $y$, it is essential to consider entangled states in the frequency domain
\begin{equation}\begin{aligned}\label{eq:psi}
&\ket{\psi}=\int dk     \Tilde{\psi}(k)\left(\psi_l\frac{(a^\dagger_{k,l})^N}{\sqrt{N!}}+\psi_r\frac{(a^\dagger_{k,r})^N}{\sqrt{N!}}\right)\ket{0},\\
\end{aligned}
\end{equation}
where the subscripts $l$ and $r$ denote the modes on the left (Alice’s side) and right (Bob’s side), respectively, and $\Tilde{\psi}(k)$ denotes the freely chosen pulse shape in the frequency domain. This state is essentially a frequency-domain NOON state. Ensuring normalization of the state requires additional justification, which is provided in Sec.~\ref{SI:precision} of the Supplemental Material. Different choices of $\ket{\psi_i}$ correspond to different $\psi_{l,r}$.

\begin{figure}[!bt]
\begin{center}
\includegraphics[width=1\columnwidth]{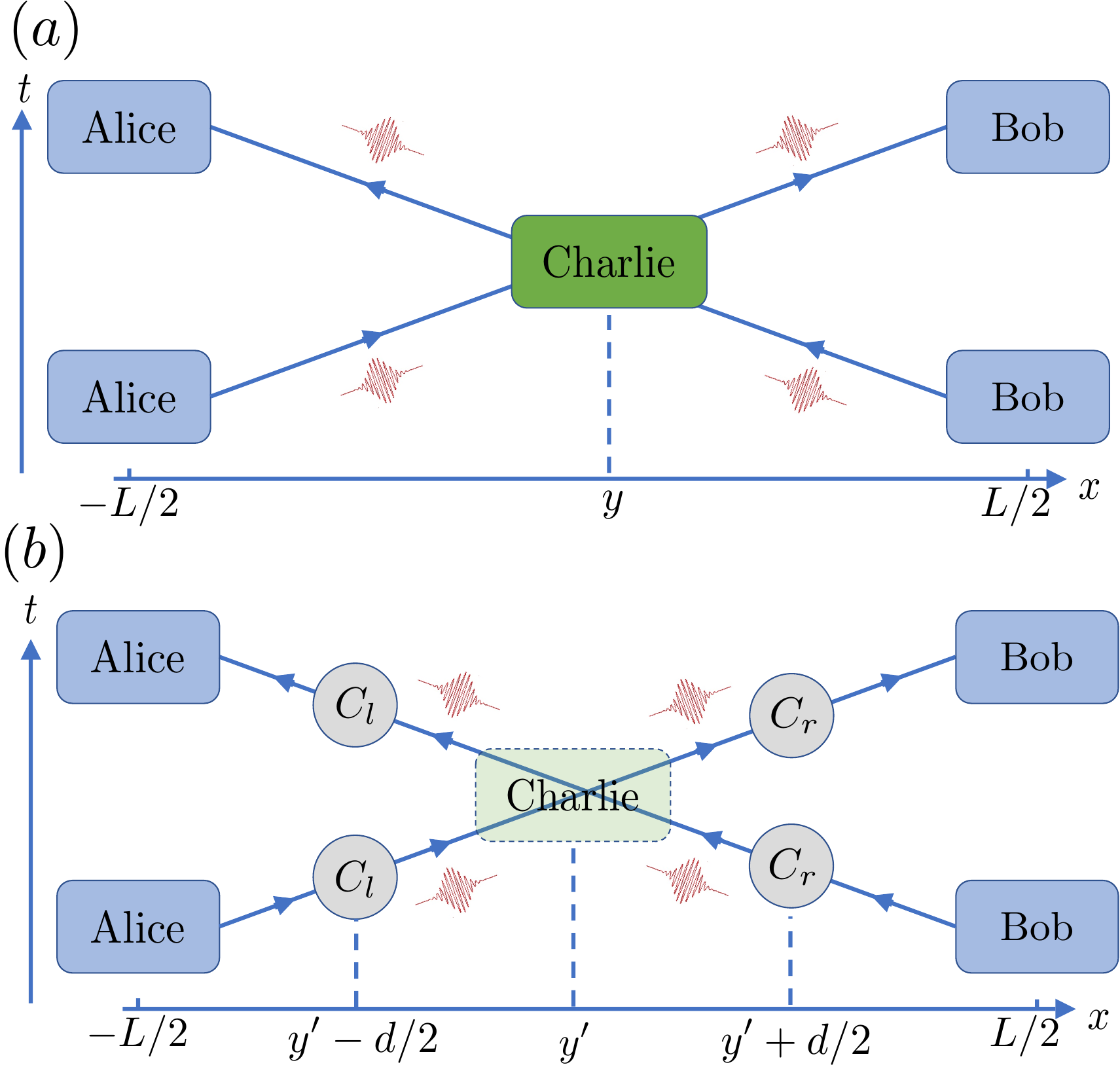}
\caption{Setup for secure quantum ranging. We consider the problem in one dimension, and the plot illustrates the propagation of the pulses in space $x$ and time $t$. (a) The honest prover Charlie, located at position $y$, performs operations on the light received from Alice and Bob. (b) Cheaters, denoted as $C_{l}$ and $C_{r}$, perform local operations at positions $y' \pm d/2$ in an attempt to impersonate Charlie at a fake position $y'$. }
\label{set_up}
\end{center}
\end{figure}

Charlie, located at position $y$, then applies a beam splitter operation to the pulses received from Alice and Bob
\begin{equation}\label{eq:U}
\begin{aligned}
&a_{y,l}^\dagger\rightarrow U_{00}a_{y,l}^\dagger+U_{10}a_{y,r}^\dagger,\,\,a_{y,r}^\dagger\rightarrow U_{01}a_{y,l}^\dagger+U_{11}a_{y,r}^\dagger.\\
\end{aligned}\end{equation}
Note that the beam splitter operation is applied to the spatial mode at position $y$ at each moment in time. As the light pulses propagate through the beam splitter, the entire pulse shape is affected by the operation. After Charlie applies the beam splitter, the resulting states is derived in  Sec.~\ref{SI:precision} of the Supplemental Material as
\begin{equation}
\begin{aligned}\label{eq:phi}
\ket{\phi}&=\frac{\psi_l}{\sqrt{N!}}\int d\vec{x}     \psi_l(\vec{x})\prod_{i=1}^N(U_{00}a_{x_i,l}^\dagger+U_{10}a_{2y-x_i,r}^\dagger)\ket{0}\\
&+\frac{\psi_r}{\sqrt{N!}}\int d\vec{x}\psi_r(\vec{x})     \prod_{i=1}^N(U_{01}a_{x_i,l}^\dagger+U_{11}a_{2y-x_i,r}^\dagger)\ket{0},\\
\end{aligned}
\end{equation}
where $\vec{x} = [x_1, x_2, \cdots, x_N]$, and the pulse shapes are given by $\psi_l(\vec{x}) \propto \psi\left(\sum_i x_i - cNt + \frac{NL}{2}\right)$ and $\psi_r(\vec{x}) \propto \psi\left(\sum_i x_i + cNt - \frac{NL}{2}\right)$, where $\psi(x)$ is the Fourier transform of $\tilde{\psi}(k)$. Given the output state $\ket{\phi}$ as a function of $y$, we can compute the quantum Fisher information (QFI), whose inverse provides a lower bound on the variance of estimating $y$ \cite{braunstein1994statistical,paris2009quantum,kay1993fundamentals}, as detailed in Sec.~\ref{SI:precision} of the Supplemental Material.
Our general expression for the QFI applies to arbitrary choices of the beam splitter operation $U_{ij}$ and pulse shape $\tilde{\psi}(k)$. However, to illustrate the achievable precision more concretely, we present the following theorem using specific choices of $U_{ij}$ and $\tilde{\psi}(k)$.

\begin{thm}\label{thm:QFI_theorem}
If Alice and Bob send the states in Eq. \ref{eq:psi} with $\tilde{\psi}(k) \propto \exp(-k^2 / 2\beta^2)$, and Charlie performs the operation given in Eq. \ref{eq:U} with
\begin{equation}\label{eq:equal_U}
U=\frac{1}{\sqrt{2}}\left[\begin{matrix}
1 & 1\\
-1 & 1
\end{matrix}\right],
\end{equation}
and with any $\psi_l,\psi_r$ and any $y$, the QFI of estimating the position of Charlie $y$ is
\begin{equation}
F=2N(N+1)\beta^2.
\end{equation}
\end{thm}
The proof of the theorem is provided in Sec.~\ref{SI:precision} of the Supplemental Material. By employing entangled probe states in the frequency domain, the QFI is shown to scale quadratically with $N$, thereby achieving Heisenberg-limited precision. This demonstrates that frequency entanglement can be effectively used to enhance the accuracy of position estimation.

Note that in typical quantum ranging protocols, such as those in Refs.~\cite{giovannetti2001quantum,maccone2020quantum}, light is sent from only one direction, and the position is estimated solely from the reflected signal. This contrasts with our setting, where light arrives from both sides and a local operation is performed at Charlie’s location. The single-sided ranging scenario can also be captured within our framework as a special case, and corresponds to the following proposition.
\begin{proposition}\label{prop:qfi_single_side}
In the single-sided ranging scenario, Alice and Bob send the states defined in Eq.~\ref{eq:psi} with $\tilde{\psi}(k) \propto \exp(-k^2 / 2\beta^2)$, $\psi_l=1$, $\psi_r=0$, and Charlie performs the operation
\begin{equation}\label{eq:U_reflection}
U=\left[\begin{matrix}
0 & 1\\
1 & 0
\end{matrix}\right].
\end{equation}
 We can calculate the QFI of estimating $y$ as
\begin{equation}
F=8\beta^2 N^2.
\end{equation}    
\end{proposition}
In this single-sided case, the QFI is approximately four times larger than that of our chosen scheme in Theorem \ref{thm:QFI_theorem} when $N$ is sufficiently large. Although this traditional approach achieves better precision by a constant factor, we will show below that it is not secure against cheating—even when adversaries are restricted to local beam splitter operations without any entanglement. In contrast, our scheme remains secure while still achieving quadratic scaling of the QFI with respect to $N$.

\textit{Security of the ranging measurement} - We now examine the security of this quantum ranging method. Following the standard setting of QPV as in Refs.~\cite{beausoleil2006tagging, malaney2010quantum, malaney2010location, kent2011quantum, buhrman2014position, vaidman2003instantaneous, beigi2011simplified, bluhm2022single, asadi2025linear,gonzales2019bounds, speelman2016instantaneous, chakraborty2015practical}, we consider two cheaters positioned at $y' \pm d/2$, where $y'$ is the fake position they attempt to attribute to Charlie, while Charlie’s actual position is $y$. Note that if the two cheaters share an unlimited amount of entanglement, then—together with local operations—they can simulate any operation that Charlie performs, as is well known in standard QPV settings \cite{vaidman2003instantaneous}. Therefore, the security of QPV protocols fundamentally relies on limiting the capabilities of the cheaters. Since in our protocol the honest prover Charlie is only required to perform a simple beam splitter operation, we impose a corresponding restriction on the cheaters’ capabilities in this discussion. We assume that the cheaters have full knowledge of the operation $U$ performed at Charlie’s location, as well as the ensemble of states used by Alice and Bob, $\{p_i, \rho_i\}$. However, they do not know which specific state $\ket{\psi_i}$ is sent in each round, since it is randomly selected by Alice and Bob. This randomness is essential: if the cheaters know in advance which state $\ket{\psi}$ would be used, they could simply discard the light from Alice and Bob and prepare the corresponding output state themselves, perfectly reproducing $\ket{\phi}$ without any interaction. To avoid this security vulnerability, Alice and Bob randomly choose the input state from the predefined ensemble $\{p_i, \ket{\psi_i}\}$.

Each cheater is allowed to introduce an ancillary mode initially in vacuum states at their location and perform any local two-port beam splitter operations. They are also equipped with quantum memory, allowing them to delay pulses arbitrarily without altering their shape.
In the first round of operations, the left and right cheaters apply beam splitter operations $V$ and $W$, respectively, on the incoming mode and their ancillary mode. Each cheater then stores their ancillary mode locally in a delay line for a duration of $d/c$. After this delay, they perform a second set of operations, $P$ and $Q$, on two modes: their stored ancillary mode and the mode received from the other side.
As detailed in Sec.~\ref{SI:state_cheater} of the Supplemental Material, the resulting state prepared by the cheaters is
\begin{equation}
\begin{aligned}
\ket{\gamma}&=\frac{\psi_l}{\sqrt{N!}}\int d\vec{x}     \psi_l(\vec{x})\prod_{k=1}^N(V_{00}P_{00}a_{x_k,l}^\dagger+V_{00}P_{01}b_{x_k,r}^\dagger\\
&\quad\quad\quad\quad+V_{01}Q_{10}a_{2y'-x_k,r}^\dagger+V_{01}Q_{11}b_{2y'-x_k,l}^\dagger)\ket{0}\\
&+\frac{\psi_r}{\sqrt{N!}}\int d\vec{x}     \psi_r(\vec{x})\prod_{k=1}^N(W_{01}P_{10}a_{x_k,l}^\dagger+W_{01}P_{11}b_{x_k,r}^\dagger\\
&\quad\quad\quad\quad+W_{00}Q_{00}a_{2y'-x_k,r}^\dagger+W_{00}Q_{01}b_{2y'-x_k,l}^\dagger)\ket{0},\\
\end{aligned}
\end{equation}
where $V_{ij}$ denotes the matrix elements of  $V$, with similar definitions for $W$, $P$, and $Q$.  The cheaters aim to optimize the operations $V$, $W$, $P$, and $Q$ such that the resulting state $\ket{\gamma}$ closely approximates the target state $\ket{\phi}$ that would have been produced by a fake Charlie located at position $y'$, as defined in Eq.~\ref{eq:phi}, where $y' \neq y$. If the cheaters can exactly reproduce the state $\ket{\phi}$ corresponding to the fake position $y'$, they can successfully spoof Charlie’s location.
Having specified the operations available to the two cheaters, we now consider two key scenarios in which they are able to perfectly simulate Charlie’s presence at the fake position.

\begin{proposition}\label{prop:reflection}
If Charlie's operation is simply the reflection given in Eq.~\ref{eq:U_reflection}, then for any $\psi_{l,r}$, $y'$, and $N$, the cheaters can always prepare $\ket{\gamma} = \ket{\phi}$.
\end{proposition}
\begin{proposition}\label{prop:single_side}
If the light is always sent from a fixed side—for instance, from Alice with $\psi_l = 1$ and $\psi_r = 0$, and there is no randomization over the input states—then, for any photon number $N$, position $y$, and unitary $U$, the cheaters can always prepare a state $\ket{\gamma} = \ket{\phi}$ by  perfectly simulating the action of $U$ on the single-sided input or discarding the original states and generating their own.
\end{proposition}

The proofs of Propositions~\ref{prop:reflection} and \ref{prop:single_side} are provided in Sec.~\ref{proof_prop:reflection} of the Supplemental Material. These results show that standard quantum ranging protocols based solely on reflecting a fixed, single-sided input—such as those in Refs. \cite{giovannetti2001quantum,maccone2020quantum,zhuang2022ultimate,zhuang2021quantum}—cannot detect the presence of potential cheaters. To ensure security, it is essential to introduce randomness in the choice of two-sided input states and to employ general operations $U$ at Charlie's location.
Although the QFI for position estimation in these single-sided protocols, as stated in Proposition~\ref{prop:qfi_single_side}, exceeds that of our two-sided ranging protocol in Theorem~\ref{thm:QFI_theorem} by a constant factor, such protocols provide no mechanism to detect cheating. In contrast, our protocol enables the detection of potential cheaters while maintaining estimation precision which achieves the Heisenberg scaling. We now proceed to analyze the error     probability of detecting cheating when general operations $U$ and random two-sided input states are used.

\begin{thm}\label{thm:optimal_gamma}
For any fake position $y'$ and a set of randomized input states from Alice and Bob, denoted by $\{p_i, \ket{\psi_i}\}$—where $\psi_{il}$ and $\psi_{ir}$ are the amplitudes of the state $\ket{\psi_i}$ at Alice’s and Bob’s input ports, respectively. The total error probability—defined as the sum of the probability of mistakenly assuming there are no cheaters when cheaters are present, and the probability of falsely assuming there are cheaters when there are none—is given by:
\begin{equation} P = \max\{P_1, P_2\}, \end{equation}
where $P_1$ is the total error probability when the cheaters attempt to apply approximate operations to the original input states, and $P_2$ is the total error probability when the cheaters discard the original states and instead prepare their own state. Note that for both $P_1$ and $P_2$, we assume that Alice and Bob perform optimal measurements to detect the presence of cheaters, while the cheaters adopt their best possible strategies in each respective case.
The error probability $P_1$ and $P_2$ can be bounded as follows:
\begin{equation}\begin{aligned}
 &P_1\leq \sum_ip_i\max\{|\psi_{il}|^2,|\psi_{ir}|^2\},\quad N\rightarrow\infty\\
 &P_2\leq \sqrt{\sum_{ij}p_i p_j |\bra{\psi_i}\ket{ \psi_j}|}.
\end{aligned}
\end{equation}
\end{thm}

The proof is provided in Sec.~\ref{proof_thm:optimal_gamma} of the Supplemental Material. Since $P_1$ and $P_2$ are defined under the assumption that both Alice and Bob, as well as the cheaters, adopt optimal strategies, one might expect a concrete value for each given set $\{p_i, \ket{\psi_i}\}$. However, we are unable to find a closed-form expression for  $P_{1,2}$.
Furthermore, the bound for $P_1$ relies on the assumption of an infinite number of photons in the entangled state, i.e., $N \to \infty$. For finite $N$, the error probability depends on the specific choices of $U$, $\psi_{l,r}$, and the cheaters’ strategies, including their local operations $W$, $V$, $P$, and $Q$. In general, optimizing over the cheaters’ strategies does not yield a closed-form solution for arbitrary $U$ and $\psi_{l,r}$. Therefore, we only provide a general expression for the error probability in Sec.~\ref{SI:state_cheater} of the Supplemental Material.

Note that the overall error probability $P$ is fundamentally bounded away from zero, since the error probability $P_2$ follows the lower bound $P_2 \geq 1 - \frac{1}{2} \sum_{i,j} p_i p_j \| \rho_i - \rho_j\|$ as shown in Sec.~\ref{proof_thm:optimal_gamma} of the Supplemental Material, where $\|.\|$ denotes the trace norm. This limitation arises because, for any ensemble of input states $\{p_i, \ket{\psi_i}\}$, the cheaters can always discard the incoming states from Alice and Bob and instead return a fixed state $\rho_0$ of their own choosing. Since $\rho_0$ can always partially approximate the ensemble $\{p_i, \ket{\psi_i}\}$, the overall error probability is lower bounded by a constant.

We emphasize that the above bound applies to the error probability of detecting cheating from a single copy of the state. However, as long as there is a constant (nonzero) probability of detecting cheaters in each round, the detection can be repeated across multiple copies of the state, eventually achieving an overall error probability arbitrarily close to zero. Such repetition is already necessary for position estimation, which requires multiple copies to establish statistical confidence. Therefore, a constant per-copy error probability is sufficient to ensure the security of the quantum ranging protocol.

\textit{Simultaneously achieving Heisenberg-limited precision and reliable cheater detection} - An important remaining question is whether the QFI in Theorem~\ref{thm:QFI_theorem} and the error probability in Theorem~\ref{thm:optimal_gamma} can be saturated by a simple measurement strategy. In the case of single-parameter estimation, the QFI bound can always be saturated, and the optimal error probability for detecting cheaters can also be achieved with an appropriate measurement \cite{braunstein1994statistical,paris2009quantum,kay1993fundamentals,holevo2011probabilistic}. However, these two optimal measurements are generally  difficult to implement in practice.
In the following Theorem~\ref{thm:simultaneous}, we present a particularly simple and physically motivated measurement strategy—based on photon counting and the use of a beam splitter—that simultaneously achieves both objectives. This method provides Heisenberg-scaling precision for position estimation while maintaining a constant error probability for detecting cheaters, all within a  unified measurement framework.

\begin{thm}\label{thm:simultaneous}
Consider an ensemble consisting of two states $\ket{\psi_1}$ and $\ket{\psi_2}$, chosen with equal probability.  Charlie always applies the unitary operation $U$ defined in Eq.~\ref{eq:equal_U}.
For detection, Alice and Bob interfere the incoming light using an additional beam splitter operation $R = U$, and then perform local single-photon measurements at each spatial mode. This projects onto states of the form $ \ket{\vec{q}, \vec{z}} = \prod_{j=1}^N b_{z_j, q_j}^\dagger \ket{0}$,
where $\vec{q} \in \{0,1\}^N$, and $b_{z_j, q_j}$ denotes the mode at position $z_j$ and output port $q_j$ of the beam splitter used in Alice and Bob's measurement.
This measurement strategy achieves a Fisher information (FI) of
\begin{equation} F = 2M\beta^2 N^2, \end{equation}
and yields an error probability of detecting the presence of cheaters that is upper bounded by
\begin{equation} P \leq \xi^M, \end{equation}
when detecting from $M$ copies of the state, after excluding certain special cases of cheaters' strategies using some overhead, where $\xi$ is a constant given in Sec.~\ref{SI:simultaneous} of the Supplemental Material.
\end{thm}

The proof is given in Sec.~\ref{SI:simultaneous} of the Supplemental Material.
Note that the FI achieved using this measurement strategy—whose inverse lower bounds the variance of estimating $y$ \cite{kay1993fundamentals}—is slightly worse than the QFI predicted in Theorem~\ref{thm:QFI_theorem}, but it still achieves Heisenberg scaling over $N$, demonstrating the advantage of using entangled states by Alice and Bob.
Note that the strategy used in Theorem~\ref{thm:simultaneous} simply requires a balanced beam splitter to interfere the light at Alice and Bob before they measure the single photon at all positions. The simple structure of such a measurement strategy should be of more practical interest.

This calculation also highlights a key limitation in standard QPV discussion: most existing works focus on qubit-based protocols, neglecting the pulse shape of the transmitted states and implicitly assuming infinite precision in verifying the prover’s position \cite{beausoleil2006tagging, malaney2010quantum, malaney2010location, kent2011quantum, buhrman2014position, vaidman2003instantaneous, beigi2011simplified, bluhm2022single, asadi2025linear, gonzales2019bounds, speelman2016instantaneous, chakraborty2015practical}. In practice, however, the states have finite temporal profiles, which fundamentally limit the achievable spatial resolution. Sharper pulses yield better precision, as a larger $\beta$ leads to higher FI/QFI $F$, and achieving infinite precision would require infinitely sharp pulses—an unphysical idealization. Interestingly, many QPV protocols produce highly entangled states, which could potentially enhance position verification accuracy. Yet, this metrological advantage has largely been overlooked. The results presented above represent a first attempt to address this overlooked aspect of QPV.

\textit{Conclusion and discussion} - In this work, we bridge the study of quantum ranging with QPV, resulting in a scheme that can estimate the position of a prover while simultaneously detecting the potential presence of cheaters. Our scheme relies on entangled states in the frequency domain and requires only a beam splitter at the prover’s location. It is secure against cheaters who possess a single ancillary mode, have access to ideal quantum memories, and are capable of performing arbitrary two-port beam splitter operations—without the use of entanglement.
This work introduces a new aspect to the study of quantum ranging—namely, the security associated with detecting cheaters who attempt to spoof the prover’s position. It also contributes a novel perspective to QPV by providing the first quantitative treatment of the precision in position verification, an aspect that is typically overlooked in existing QPV protocols.

This work initiates a new line of inquiry at the intersection of quantum ranging and QPV, opening up a broad range of potential directions for future exploration.
We rely on single-mode states from Alice and Bob and require only a simple beam splitter operation from Charlie, which enhances the scheme's practicality. However, it is natural to ask whether more sophisticated constructions—such as multimode states or operations involving nonlinear optical devices—could further improve the scheme’s security or precision.
Another interesting direction is to explore whether entanglement between the probe states and an idler mode, kept locally at Alice and Bob, could enhance performance, similar to ideas developed in the context of quantum illumination~\cite{tan2008quantum,lloyd2008enhanced,barzanjeh2015microwave,nair2020fundamental,sanz2017quantum,barzanjeh2020microwave,gregory2020imaging,karsa2024quantum,shapiro2020quantum}. Such entanglement might eliminate the need for randomness in the states sent by Alice and Bob, and could also offer additional robustness against noise and loss.

\textit{Acknowledgements} - We thank Yujie Zhang, Peixue Wu, Debbie Leung for helpful discussions. AM and YW acknowledge the support of the Natural Sciences and Engineering Research Council of Canada (NSERC); this work was supported by an NSERC-UKRI Alliance grant (ALLRP 597823-24).  GS and YW acknowledge funding from the Canada First Research Excellence Fund.

\bibliography{main}

\begin{thebibliography}{32}%
\makeatletter
\providecommand \@ifxundefined [1]{%
 \@ifx{#1\undefined}
}%
\providecommand \@ifnum [1]{%
 \ifnum #1\expandafter \@firstoftwo
 \else \expandafter \@secondoftwo
 \fi
}%
\providecommand \@ifx [1]{%
 \ifx #1\expandafter \@firstoftwo
 \else \expandafter \@secondoftwo
 \fi
}%
\providecommand \natexlab [1]{#1}%
\providecommand \enquote  [1]{``#1''}%
\providecommand \bibnamefont  [1]{#1}%
\providecommand \bibfnamefont [1]{#1}%
\providecommand \citenamefont [1]{#1}%
\providecommand \href@noop [0]{\@secondoftwo}%
\providecommand \href [0]{\begingroup \@sanitize@url \@href}%
\providecommand \@href[1]{\@@startlink{#1}\@@href}%
\providecommand \@@href[1]{\endgroup#1\@@endlink}%
\providecommand \@sanitize@url [0]{\catcode `\\12\catcode `\$12\catcode
  `\&12\catcode `\#12\catcode `\^12\catcode `\_12\catcode `\%12\relax}%
\providecommand \@@startlink[1]{}%
\providecommand \@@endlink[0]{}%
\providecommand \url  [0]{\begingroup\@sanitize@url \@url }%
\providecommand \@url [1]{\endgroup\@href {#1}{\urlprefix }}%
\providecommand \urlprefix  [0]{URL }%
\providecommand \Eprint [0]{\href }%
\providecommand \doibase [0]{https://doi.org/}%
\providecommand \selectlanguage [0]{\@gobble}%
\providecommand \bibinfo  [0]{\@secondoftwo}%
\providecommand \bibfield  [0]{\@secondoftwo}%
\providecommand \translation [1]{[#1]}%
\providecommand \BibitemOpen [0]{}%
\providecommand \bibitemStop [0]{}%
\providecommand \bibitemNoStop [0]{.\EOS\space}%
\providecommand \EOS [0]{\spacefactor3000\relax}%
\providecommand \BibitemShut  [1]{\csname bibitem#1\endcsname}%
\let\auto@bib@innerbib\@empty
%</preamble>
\bibitem [{\citenamefont {Giovannetti}\ \emph {et~al.}(2001)\citenamefont
  {Giovannetti}, \citenamefont {Lloyd},\ and\ \citenamefont
  {Maccone}}]{giovannetti2001quantum}%
  \BibitemOpen
  \bibfield  {author} {\bibinfo {author} {\bibfnamefont {V.}~\bibnamefont
  {Giovannetti}}, \bibinfo {author} {\bibfnamefont {S.}~\bibnamefont {Lloyd}},\
  and\ \bibinfo {author} {\bibfnamefont {L.}~\bibnamefont {Maccone}},\
  }\bibfield  {title} {\bibinfo {title} {Quantum-enhanced positioning and clock
  synchronization},\ }\href@noop {} {\bibfield  {journal} {\bibinfo  {journal}
  {Nature}\ }\textbf {\bibinfo {volume} {412}},\ \bibinfo {pages} {417}
  (\bibinfo {year} {2001})}\BibitemShut {NoStop}%
\bibitem [{\citenamefont {Maccone}\ and\ \citenamefont
  {Ren}(2020)}]{maccone2020quantum}%
  \BibitemOpen
  \bibfield  {author} {\bibinfo {author} {\bibfnamefont {L.}~\bibnamefont
  {Maccone}}\ and\ \bibinfo {author} {\bibfnamefont {C.}~\bibnamefont {Ren}},\
  }\bibfield  {title} {\bibinfo {title} {Quantum radar},\ }\href@noop {}
  {\bibfield  {journal} {\bibinfo  {journal} {Physical Review Letters}\
  }\textbf {\bibinfo {volume} {124}},\ \bibinfo {pages} {200503} (\bibinfo
  {year} {2020})}\BibitemShut {NoStop}%
\bibitem [{\citenamefont {Zhuang}\ and\ \citenamefont
  {Shapiro}(2022)}]{zhuang2022ultimate}%
  \BibitemOpen
  \bibfield  {author} {\bibinfo {author} {\bibfnamefont {Q.}~\bibnamefont
  {Zhuang}}\ and\ \bibinfo {author} {\bibfnamefont {J.~H.}\ \bibnamefont
  {Shapiro}},\ }\bibfield  {title} {\bibinfo {title} {Ultimate accuracy limit
  of quantum pulse-compression ranging},\ }\href@noop {} {\bibfield  {journal}
  {\bibinfo  {journal} {Physical review letters}\ }\textbf {\bibinfo {volume}
  {128}},\ \bibinfo {pages} {010501} (\bibinfo {year} {2022})}\BibitemShut
  {NoStop}%
\bibitem [{\citenamefont {Zhuang}(2021)}]{zhuang2021quantum}%
  \BibitemOpen
  \bibfield  {author} {\bibinfo {author} {\bibfnamefont {Q.}~\bibnamefont
  {Zhuang}},\ }\bibfield  {title} {\bibinfo {title} {Quantum ranging with
  gaussian entanglement},\ }\href@noop {} {\bibfield  {journal} {\bibinfo
  {journal} {Physical Review Letters}\ }\textbf {\bibinfo {volume} {126}},\
  \bibinfo {pages} {240501} (\bibinfo {year} {2021})}\BibitemShut {NoStop}%
\bibitem [{\citenamefont {Tan}\ \emph {et~al.}(2008)\citenamefont {Tan},
  \citenamefont {Erkmen}, \citenamefont {Giovannetti}, \citenamefont {Guha},
  \citenamefont {Lloyd}, \citenamefont {Maccone}, \citenamefont {Pirandola},\
  and\ \citenamefont {Shapiro}}]{tan2008quantum}%
  \BibitemOpen
  \bibfield  {author} {\bibinfo {author} {\bibfnamefont {S.-H.}\ \bibnamefont
  {Tan}}, \bibinfo {author} {\bibfnamefont {B.~I.}\ \bibnamefont {Erkmen}},
  \bibinfo {author} {\bibfnamefont {V.}~\bibnamefont {Giovannetti}}, \bibinfo
  {author} {\bibfnamefont {S.}~\bibnamefont {Guha}}, \bibinfo {author}
  {\bibfnamefont {S.}~\bibnamefont {Lloyd}}, \bibinfo {author} {\bibfnamefont
  {L.}~\bibnamefont {Maccone}}, \bibinfo {author} {\bibfnamefont
  {S.}~\bibnamefont {Pirandola}},\ and\ \bibinfo {author} {\bibfnamefont
  {J.~H.}\ \bibnamefont {Shapiro}},\ }\bibfield  {title} {\bibinfo {title}
  {Quantum illumination with gaussian states},\ }\href@noop {} {\bibfield
  {journal} {\bibinfo  {journal} {Physical review letters}\ }\textbf {\bibinfo
  {volume} {101}},\ \bibinfo {pages} {253601} (\bibinfo {year}
  {2008})}\BibitemShut {NoStop}%
\bibitem [{\citenamefont {Lloyd}(2008)}]{lloyd2008enhanced}%
  \BibitemOpen
  \bibfield  {author} {\bibinfo {author} {\bibfnamefont {S.}~\bibnamefont
  {Lloyd}},\ }\bibfield  {title} {\bibinfo {title} {Enhanced sensitivity of
  photodetection via quantum illumination},\ }\href@noop {} {\bibfield
  {journal} {\bibinfo  {journal} {Science}\ }\textbf {\bibinfo {volume}
  {321}},\ \bibinfo {pages} {1463} (\bibinfo {year} {2008})}\BibitemShut
  {NoStop}%
\bibitem [{\citenamefont {Barzanjeh}\ \emph {et~al.}(2015)\citenamefont
  {Barzanjeh}, \citenamefont {Guha}, \citenamefont {Weedbrook}, \citenamefont
  {Vitali}, \citenamefont {Shapiro},\ and\ \citenamefont
  {Pirandola}}]{barzanjeh2015microwave}%
  \BibitemOpen
  \bibfield  {author} {\bibinfo {author} {\bibfnamefont {S.}~\bibnamefont
  {Barzanjeh}}, \bibinfo {author} {\bibfnamefont {S.}~\bibnamefont {Guha}},
  \bibinfo {author} {\bibfnamefont {C.}~\bibnamefont {Weedbrook}}, \bibinfo
  {author} {\bibfnamefont {D.}~\bibnamefont {Vitali}}, \bibinfo {author}
  {\bibfnamefont {J.~H.}\ \bibnamefont {Shapiro}},\ and\ \bibinfo {author}
  {\bibfnamefont {S.}~\bibnamefont {Pirandola}},\ }\bibfield  {title} {\bibinfo
  {title} {Microwave quantum illumination},\ }\href@noop {} {\bibfield
  {journal} {\bibinfo  {journal} {Physical review letters}\ }\textbf {\bibinfo
  {volume} {114}},\ \bibinfo {pages} {080503} (\bibinfo {year}
  {2015})}\BibitemShut {NoStop}%
\bibitem [{\citenamefont {Nair}\ and\ \citenamefont
  {Gu}(2020)}]{nair2020fundamental}%
  \BibitemOpen
  \bibfield  {author} {\bibinfo {author} {\bibfnamefont {R.}~\bibnamefont
  {Nair}}\ and\ \bibinfo {author} {\bibfnamefont {M.}~\bibnamefont {Gu}},\
  }\bibfield  {title} {\bibinfo {title} {Fundamental limits of quantum
  illumination},\ }\href@noop {} {\bibfield  {journal} {\bibinfo  {journal}
  {Optica}\ }\textbf {\bibinfo {volume} {7}},\ \bibinfo {pages} {771} (\bibinfo
  {year} {2020})}\BibitemShut {NoStop}%
\bibitem [{\citenamefont {Sanz}\ \emph {et~al.}(2017)\citenamefont {Sanz},
  \citenamefont {Las~Heras}, \citenamefont {Garc{\'\i}a-Ripoll}, \citenamefont
  {Solano},\ and\ \citenamefont {Di~Candia}}]{sanz2017quantum}%
  \BibitemOpen
  \bibfield  {author} {\bibinfo {author} {\bibfnamefont {M.}~\bibnamefont
  {Sanz}}, \bibinfo {author} {\bibfnamefont {U.}~\bibnamefont {Las~Heras}},
  \bibinfo {author} {\bibfnamefont {J.~J.}\ \bibnamefont {Garc{\'\i}a-Ripoll}},
  \bibinfo {author} {\bibfnamefont {E.}~\bibnamefont {Solano}},\ and\ \bibinfo
  {author} {\bibfnamefont {R.}~\bibnamefont {Di~Candia}},\ }\bibfield  {title}
  {\bibinfo {title} {Quantum estimation methods for quantum illumination},\
  }\href@noop {} {\bibfield  {journal} {\bibinfo  {journal} {Physical review
  letters}\ }\textbf {\bibinfo {volume} {118}},\ \bibinfo {pages} {070803}
  (\bibinfo {year} {2017})}\BibitemShut {NoStop}%
\bibitem [{\citenamefont {Barzanjeh}\ \emph {et~al.}(2020)\citenamefont
  {Barzanjeh}, \citenamefont {Pirandola}, \citenamefont {Vitali},\ and\
  \citenamefont {Fink}}]{barzanjeh2020microwave}%
  \BibitemOpen
  \bibfield  {author} {\bibinfo {author} {\bibfnamefont {S.}~\bibnamefont
  {Barzanjeh}}, \bibinfo {author} {\bibfnamefont {S.}~\bibnamefont
  {Pirandola}}, \bibinfo {author} {\bibfnamefont {D.}~\bibnamefont {Vitali}},\
  and\ \bibinfo {author} {\bibfnamefont {J.~M.}\ \bibnamefont {Fink}},\
  }\bibfield  {title} {\bibinfo {title} {Microwave quantum illumination using a
  digital receiver},\ }\href@noop {} {\bibfield  {journal} {\bibinfo  {journal}
  {Science advances}\ }\textbf {\bibinfo {volume} {6}},\ \bibinfo {pages}
  {eabb0451} (\bibinfo {year} {2020})}\BibitemShut {NoStop}%
\bibitem [{\citenamefont {Gregory}\ \emph {et~al.}(2020)\citenamefont
  {Gregory}, \citenamefont {Moreau}, \citenamefont {Toninelli},\ and\
  \citenamefont {Padgett}}]{gregory2020imaging}%
  \BibitemOpen
  \bibfield  {author} {\bibinfo {author} {\bibfnamefont {T.}~\bibnamefont
  {Gregory}}, \bibinfo {author} {\bibfnamefont {P.-A.}\ \bibnamefont {Moreau}},
  \bibinfo {author} {\bibfnamefont {E.}~\bibnamefont {Toninelli}},\ and\
  \bibinfo {author} {\bibfnamefont {M.~J.}\ \bibnamefont {Padgett}},\
  }\bibfield  {title} {\bibinfo {title} {Imaging through noise with quantum
  illumination},\ }\href@noop {} {\bibfield  {journal} {\bibinfo  {journal}
  {Science advances}\ }\textbf {\bibinfo {volume} {6}},\ \bibinfo {pages}
  {eaay2652} (\bibinfo {year} {2020})}\BibitemShut {NoStop}%
\bibitem [{\citenamefont {Karsa}\ \emph {et~al.}(2024)\citenamefont {Karsa},
  \citenamefont {Fletcher}, \citenamefont {Spedalieri},\ and\ \citenamefont
  {Pirandola}}]{karsa2024quantum}%
  \BibitemOpen
  \bibfield  {author} {\bibinfo {author} {\bibfnamefont {A.}~\bibnamefont
  {Karsa}}, \bibinfo {author} {\bibfnamefont {A.}~\bibnamefont {Fletcher}},
  \bibinfo {author} {\bibfnamefont {G.}~\bibnamefont {Spedalieri}},\ and\
  \bibinfo {author} {\bibfnamefont {S.}~\bibnamefont {Pirandola}},\ }\bibfield
  {title} {\bibinfo {title} {Quantum illumination and quantum radar: A brief
  overview},\ }\href@noop {} {\bibfield  {journal} {\bibinfo  {journal}
  {Reports on progress in physics}\ }\textbf {\bibinfo {volume} {87}},\
  \bibinfo {pages} {094001} (\bibinfo {year} {2024})}\BibitemShut {NoStop}%
\bibitem [{\citenamefont {Shapiro}(2020)}]{shapiro2020quantum}%
  \BibitemOpen
  \bibfield  {author} {\bibinfo {author} {\bibfnamefont {J.~H.}\ \bibnamefont
  {Shapiro}},\ }\bibfield  {title} {\bibinfo {title} {The quantum illumination
  story},\ }\href@noop {} {\bibfield  {journal} {\bibinfo  {journal} {IEEE
  Aerospace and Electronic Systems Magazine}\ }\textbf {\bibinfo {volume}
  {35}},\ \bibinfo {pages} {8} (\bibinfo {year} {2020})}\BibitemShut {NoStop}%
\bibitem [{\citenamefont {Beausoleil}\ \emph {et~al.}(2006)\citenamefont
  {Beausoleil}, \citenamefont {Kent}, \citenamefont {Munro},\ and\
  \citenamefont {Spiller}}]{beausoleil2006tagging}%
  \BibitemOpen
  \bibfield  {author} {\bibinfo {author} {\bibfnamefont {R.~G.}\ \bibnamefont
  {Beausoleil}}, \bibinfo {author} {\bibfnamefont {A.}~\bibnamefont {Kent}},
  \bibinfo {author} {\bibfnamefont {W.~J.}\ \bibnamefont {Munro}},\ and\
  \bibinfo {author} {\bibfnamefont {T.~P.}\ \bibnamefont {Spiller}},\
  }\href@noop {} {\bibinfo {title} {Tagging systems}} (\bibinfo {year}
  {2006})\BibitemShut {NoStop}%
\bibitem [{\citenamefont {Malaney}(2010{\natexlab{a}})}]{malaney2010quantum}%
  \BibitemOpen
  \bibfield  {author} {\bibinfo {author} {\bibfnamefont {R.~A.}\ \bibnamefont
  {Malaney}},\ }\bibfield  {title} {\bibinfo {title} {Quantum location
  verification in noisy channels},\ }in\ \href@noop {} {\emph {\bibinfo
  {booktitle} {2010 IEEE global telecommunications conference GLOBECOM 2010}}}\
  (\bibinfo {organization} {IEEE},\ \bibinfo {year} {2010})\ pp.\ \bibinfo
  {pages} {1--6}\BibitemShut {NoStop}%
\bibitem [{\citenamefont {Malaney}(2010{\natexlab{b}})}]{malaney2010location}%
  \BibitemOpen
  \bibfield  {author} {\bibinfo {author} {\bibfnamefont {R.~A.}\ \bibnamefont
  {Malaney}},\ }\bibfield  {title} {\bibinfo {title} {Location-dependent
  communications using quantum entanglement},\ }\href@noop {} {\bibfield
  {journal} {\bibinfo  {journal} {Physical Review A—Atomic, Molecular, and
  Optical Physics}\ }\textbf {\bibinfo {volume} {81}},\ \bibinfo {pages}
  {042319} (\bibinfo {year} {2010}{\natexlab{b}})}\BibitemShut {NoStop}%
\bibitem [{\citenamefont {Kent}\ \emph {et~al.}(2011)\citenamefont {Kent},
  \citenamefont {Munro},\ and\ \citenamefont {Spiller}}]{kent2011quantum}%
  \BibitemOpen
  \bibfield  {author} {\bibinfo {author} {\bibfnamefont {A.}~\bibnamefont
  {Kent}}, \bibinfo {author} {\bibfnamefont {W.~J.}\ \bibnamefont {Munro}},\
  and\ \bibinfo {author} {\bibfnamefont {T.~P.}\ \bibnamefont {Spiller}},\
  }\bibfield  {title} {\bibinfo {title} {Quantum tagging: Authenticating
  location via quantum information and relativistic signaling constraints},\
  }\href@noop {} {\bibfield  {journal} {\bibinfo  {journal} {Physical Review
  A—Atomic, Molecular, and Optical Physics}\ }\textbf {\bibinfo {volume}
  {84}},\ \bibinfo {pages} {012326} (\bibinfo {year} {2011})}\BibitemShut
  {NoStop}%
\bibitem [{\citenamefont {Buhrman}\ \emph {et~al.}(2014)\citenamefont
  {Buhrman}, \citenamefont {Chandran}, \citenamefont {Fehr}, \citenamefont
  {Gelles}, \citenamefont {Goyal}, \citenamefont {Ostrovsky},\ and\
  \citenamefont {Schaffner}}]{buhrman2014position}%
  \BibitemOpen
  \bibfield  {author} {\bibinfo {author} {\bibfnamefont {H.}~\bibnamefont
  {Buhrman}}, \bibinfo {author} {\bibfnamefont {N.}~\bibnamefont {Chandran}},
  \bibinfo {author} {\bibfnamefont {S.}~\bibnamefont {Fehr}}, \bibinfo {author}
  {\bibfnamefont {R.}~\bibnamefont {Gelles}}, \bibinfo {author} {\bibfnamefont
  {V.}~\bibnamefont {Goyal}}, \bibinfo {author} {\bibfnamefont
  {R.}~\bibnamefont {Ostrovsky}},\ and\ \bibinfo {author} {\bibfnamefont
  {C.}~\bibnamefont {Schaffner}},\ }\bibfield  {title} {\bibinfo {title}
  {Position-based quantum cryptography: Impossibility and constructions},\
  }\href@noop {} {\bibfield  {journal} {\bibinfo  {journal} {SIAM Journal on
  Computing}\ }\textbf {\bibinfo {volume} {43}},\ \bibinfo {pages} {150}
  (\bibinfo {year} {2014})}\BibitemShut {NoStop}%
\bibitem [{\citenamefont {Vaidman}(2003)}]{vaidman2003instantaneous}%
  \BibitemOpen
  \bibfield  {author} {\bibinfo {author} {\bibfnamefont {L.}~\bibnamefont
  {Vaidman}},\ }\bibfield  {title} {\bibinfo {title} {Instantaneous measurement
  of nonlocal variables},\ }\href@noop {} {\bibfield  {journal} {\bibinfo
  {journal} {Physical review letters}\ }\textbf {\bibinfo {volume} {90}},\
  \bibinfo {pages} {010402} (\bibinfo {year} {2003})}\BibitemShut {NoStop}%
\bibitem [{\citenamefont {Beigi}\ and\ \citenamefont
  {K{\"o}nig}(2011)}]{beigi2011simplified}%
  \BibitemOpen
  \bibfield  {author} {\bibinfo {author} {\bibfnamefont {S.}~\bibnamefont
  {Beigi}}\ and\ \bibinfo {author} {\bibfnamefont {R.}~\bibnamefont
  {K{\"o}nig}},\ }\bibfield  {title} {\bibinfo {title} {Simplified
  instantaneous non-local quantum computation with applications to
  position-based cryptography},\ }\href@noop {} {\bibfield  {journal} {\bibinfo
   {journal} {New Journal of Physics}\ }\textbf {\bibinfo {volume} {13}},\
  \bibinfo {pages} {093036} (\bibinfo {year} {2011})}\BibitemShut {NoStop}%
\bibitem [{\citenamefont {Bluhm}\ \emph {et~al.}(2022)\citenamefont {Bluhm},
  \citenamefont {Christandl},\ and\ \citenamefont
  {Speelman}}]{bluhm2022single}%
  \BibitemOpen
  \bibfield  {author} {\bibinfo {author} {\bibfnamefont {A.}~\bibnamefont
  {Bluhm}}, \bibinfo {author} {\bibfnamefont {M.}~\bibnamefont {Christandl}},\
  and\ \bibinfo {author} {\bibfnamefont {F.}~\bibnamefont {Speelman}},\
  }\bibfield  {title} {\bibinfo {title} {A single-qubit position verification
  protocol that is secure against multi-qubit attacks},\ }\href@noop {}
  {\bibfield  {journal} {\bibinfo  {journal} {Nature Physics}\ }\textbf
  {\bibinfo {volume} {18}},\ \bibinfo {pages} {623} (\bibinfo {year}
  {2022})}\BibitemShut {NoStop}%
\bibitem [{\citenamefont {Asadi}\ \emph {et~al.}(2025)\citenamefont {Asadi},
  \citenamefont {Cleve}, \citenamefont {Culf},\ and\ \citenamefont
  {May}}]{asadi2025linear}%
  \BibitemOpen
  \bibfield  {author} {\bibinfo {author} {\bibfnamefont {V.}~\bibnamefont
  {Asadi}}, \bibinfo {author} {\bibfnamefont {R.}~\bibnamefont {Cleve}},
  \bibinfo {author} {\bibfnamefont {E.}~\bibnamefont {Culf}},\ and\ \bibinfo
  {author} {\bibfnamefont {A.}~\bibnamefont {May}},\ }\bibfield  {title}
  {\bibinfo {title} {Linear gate bounds against natural functions for
  position-verification},\ }\href@noop {} {\bibfield  {journal} {\bibinfo
  {journal} {Quantum}\ }\textbf {\bibinfo {volume} {9}},\ \bibinfo {pages}
  {1604} (\bibinfo {year} {2025})}\BibitemShut {NoStop}%
\bibitem [{\citenamefont {Gonzales}\ and\ \citenamefont
  {Chitambar}(2019)}]{gonzales2019bounds}%
  \BibitemOpen
  \bibfield  {author} {\bibinfo {author} {\bibfnamefont {A.}~\bibnamefont
  {Gonzales}}\ and\ \bibinfo {author} {\bibfnamefont {E.}~\bibnamefont
  {Chitambar}},\ }\bibfield  {title} {\bibinfo {title} {Bounds on instantaneous
  nonlocal quantum computation},\ }\href@noop {} {\bibfield  {journal}
  {\bibinfo  {journal} {IEEE Transactions on Information Theory}\ }\textbf
  {\bibinfo {volume} {66}},\ \bibinfo {pages} {2951} (\bibinfo {year}
  {2019})}\BibitemShut {NoStop}%
\bibitem [{\citenamefont {Speelman}(2016)}]{speelman2016instantaneous}%
  \BibitemOpen
  \bibfield  {author} {\bibinfo {author} {\bibfnamefont {F.}~\bibnamefont
  {Speelman}},\ }\bibfield  {title} {\bibinfo {title} {Instantaneous non-local
  computation of low t-depth quantum circuits},\ }in\ \href@noop {} {\emph
  {\bibinfo {booktitle} {11th Conference on the Theory of Quantum Computation,
  Communication and Cryptography}}}\ (\bibinfo {year} {2016})\BibitemShut
  {NoStop}%
\bibitem [{\citenamefont {Chakraborty}\ and\ \citenamefont
  {Leverrier}(2015)}]{chakraborty2015practical}%
  \BibitemOpen
  \bibfield  {author} {\bibinfo {author} {\bibfnamefont {K.}~\bibnamefont
  {Chakraborty}}\ and\ \bibinfo {author} {\bibfnamefont {A.}~\bibnamefont
  {Leverrier}},\ }\bibfield  {title} {\bibinfo {title} {Practical
  position-based quantum cryptography},\ }\href@noop {} {\bibfield  {journal}
  {\bibinfo  {journal} {Physical Review A}\ }\textbf {\bibinfo {volume} {92}},\
  \bibinfo {pages} {052304} (\bibinfo {year} {2015})}\BibitemShut {NoStop}%
\bibitem [{\citenamefont {Tomamichel}\ \emph {et~al.}(2013)\citenamefont
  {Tomamichel}, \citenamefont {Fehr}, \citenamefont {Kaniewski},\ and\
  \citenamefont {Wehner}}]{tomamichel2013monogamy}%
  \BibitemOpen
  \bibfield  {author} {\bibinfo {author} {\bibfnamefont {M.}~\bibnamefont
  {Tomamichel}}, \bibinfo {author} {\bibfnamefont {S.}~\bibnamefont {Fehr}},
  \bibinfo {author} {\bibfnamefont {J.}~\bibnamefont {Kaniewski}},\ and\
  \bibinfo {author} {\bibfnamefont {S.}~\bibnamefont {Wehner}},\ }\bibfield
  {title} {\bibinfo {title} {A monogamy-of-entanglement game with applications
  to device-independent quantum cryptography},\ }\href@noop {} {\bibfield
  {journal} {\bibinfo  {journal} {New Journal of Physics}\ }\textbf {\bibinfo
  {volume} {15}},\ \bibinfo {pages} {103002} (\bibinfo {year}
  {2013})}\BibitemShut {NoStop}%
\bibitem [{\citenamefont {Braunstein}\ and\ \citenamefont
  {Caves}(1994)}]{braunstein1994statistical}%
  \BibitemOpen
  \bibfield  {author} {\bibinfo {author} {\bibfnamefont {S.~L.}\ \bibnamefont
  {Braunstein}}\ and\ \bibinfo {author} {\bibfnamefont {C.~M.}\ \bibnamefont
  {Caves}},\ }\bibfield  {title} {\bibinfo {title} {Statistical distance and
  the geometry of quantum states},\ }\href@noop {} {\bibfield  {journal}
  {\bibinfo  {journal} {Physical Review Letters}\ }\textbf {\bibinfo {volume}
  {72}},\ \bibinfo {pages} {3439} (\bibinfo {year} {1994})}\BibitemShut
  {NoStop}%
\bibitem [{\citenamefont {Paris}(2009)}]{paris2009quantum}%
  \BibitemOpen
  \bibfield  {author} {\bibinfo {author} {\bibfnamefont {M.~G.}\ \bibnamefont
  {Paris}},\ }\bibfield  {title} {\bibinfo {title} {Quantum estimation for
  quantum technology},\ }\href@noop {} {\bibfield  {journal} {\bibinfo
  {journal} {International Journal of Quantum Information}\ }\textbf {\bibinfo
  {volume} {7}},\ \bibinfo {pages} {125} (\bibinfo {year} {2009})}\BibitemShut
  {NoStop}%
\bibitem [{\citenamefont {Kay}(1993)}]{kay1993fundamentals}%
  \BibitemOpen
  \bibfield  {author} {\bibinfo {author} {\bibfnamefont {S.~M.}\ \bibnamefont
  {Kay}},\ }\href@noop {} {\emph {\bibinfo {title} {Fundamentals of statistical
  signal processing: estimation theory}}}\ (\bibinfo  {publisher}
  {Prentice-Hall, Inc.},\ \bibinfo {year} {1993})\BibitemShut {NoStop}%
\bibitem [{\citenamefont {Holevo}(2011)}]{holevo2011probabilistic}%
  \BibitemOpen
  \bibfield  {author} {\bibinfo {author} {\bibfnamefont {A.~S.}\ \bibnamefont
  {Holevo}},\ }\href@noop {} {\emph {\bibinfo {title} {Probabilistic and
  statistical aspects of quantum theory}}},\ Vol.~\bibinfo {volume} {1}\
  (\bibinfo  {publisher} {Springer Science \& Business Media},\ \bibinfo {year}
  {2011})\BibitemShut {NoStop}%
\bibitem [{\citenamefont {Sidhu}\ and\ \citenamefont
  {Kok}(2020)}]{sidhu2020geometric}%
  \BibitemOpen
  \bibfield  {author} {\bibinfo {author} {\bibfnamefont {J.~S.}\ \bibnamefont
  {Sidhu}}\ and\ \bibinfo {author} {\bibfnamefont {P.}~\bibnamefont {Kok}},\
  }\bibfield  {title} {\bibinfo {title} {Geometric perspective on quantum
  parameter estimation},\ }\href@noop {} {\bibfield  {journal} {\bibinfo
  {journal} {AVS Quantum Science}\ }\textbf {\bibinfo {volume} {2}} (\bibinfo
  {year} {2020})}\BibitemShut {NoStop}%
\bibitem [{\citenamefont {Afham}\ \emph {et~al.}(2022)\citenamefont {Afham},
  \citenamefont {Kueng},\ and\ \citenamefont {Ferrie}}]{afham2022quantum}%
  \BibitemOpen
  \bibfield  {author} {\bibinfo {author} {\bibfnamefont {A.}~\bibnamefont
  {Afham}}, \bibinfo {author} {\bibfnamefont {R.}~\bibnamefont {Kueng}},\ and\
  \bibinfo {author} {\bibfnamefont {C.}~\bibnamefont {Ferrie}},\ }\bibfield
  {title} {\bibinfo {title} {Quantum mean states are nicer than you think: Fast
  algorithms to compute states maximizing average fidelity},\ }\href@noop {}
  {\bibfield  {journal} {\bibinfo  {journal} {arXiv preprint arXiv:2206.08183}\
  } (\bibinfo {year} {2022})}\BibitemShut {NoStop}%
\end{thebibliography}%
\newpage
\onecolumngrid
\appendix

\section{Precision of quantum ranging}\label{SI:precision}

In this section, we detail the quantification of the variance in estimating Charlie's position based on quantum estimation theory \cite{braunstein1994statistical,paris2009quantum,kay1993fundamentals,sidhu2020geometric}.
To ensure proper normalization, we choose the state as follows
\begin{equation}\begin{aligned}
&\ket{\psi}=\frac{\psi_l}{\sqrt{N!}}\int d\vec{k}     \Tilde{\psi}(\vec{k})\prod_{i=1}^Na^\dagger_{k_i,l}\ket{0}+\frac{\psi_r}{\sqrt{N!}}\int d\vec{k}     \Tilde{\psi}(\vec{k})\prod_{i=1}^Na^\dagger_{k_i,r}\ket{0},\\
&\tilde{\psi}(\vec{k})=\int \tilde{\psi}(k)\frac{\exp(-\sum_{i=1}^N(k_i-k)^2/(2\sigma^2))}{(\pi\sigma^2)^{N/4}}dk,
\end{aligned}
\end{equation}
where, $\vec{k} = [k_1, k_2, \ldots, k_N]$, and the subscripts $l$ and $r$ denote the modes on the left (Alice's side) and right (Bob's side) respectively, $|\psi_l|^2+|\psi_r|^2=1$. The function $\Tilde{\psi}(\vec{k})$, which represents the pulse shape in frequency space, can be freely chosen. This state is essentially the same NOON state in frequency space as discussed in the main text, but with additional factors introduced in $\psi(\vec{k})$ to ensure proper normalization. We will eventually take the limit $\sigma \to 0$.

We convert the state to the spatial coordinate representation using $a_{k,l}=\frac{1}{\sqrt{2\pi}}\int dx e^{ik(x+\frac{L}{2}-ct)}a_{x,l}$ and $a_{k,r}=\frac{1}{\sqrt{2\pi}}\int dx e^{ik(x-\frac{L}{2}+ct)}a_{x,r}$, which gives
\begin{equation}\label{eq:SM_psi}
\begin{aligned}
\ket{\psi}&=\frac{\psi_l}{\sqrt{N!}}\int d\vec{x}     \psi\left(\sum_i x_i-cNt+\frac{NL}{2}\right)\left(\frac{\sigma^2}{\pi}\right)^{N/4}\exp\left[-\frac{\sigma^2}{2}\sum_{i=1}^N(x_i-ct+\frac{L}{2})^2\right]\prod_{i=1}^Na_{x_i,l}^\dagger\ket{0}\\
&+\frac{\psi_r}{\sqrt{N!}}\int d\vec{x}     \psi\left(\sum_i x_i+cNt-\frac{NL}{2}\right)\left(\frac{\sigma^2}{\pi}\right)^{N/4}\exp\left[-\frac{\sigma^2}{2}\sum_{i=1}^N(x_i+ct-\frac{L}{2})^2\right]\prod_{i=1}^Na_{x_i,r}^\dagger\ket{0},\\
&\psi(x)=\int dk\tilde{\psi}(k)\exp(-ikx),
\end{aligned}
\end{equation}
where, $d\vec{x} = dx_1,dx_2,\cdots,dx_N$, and the effect of pulse propagation over time has been taken into account.

Starting from the initial state $\ket{\psi}$, we consider the terms on the left from Alice at time $\tau$ and positions $\vec{x}$
\begin{equation}
\begin{aligned}
&\frac{\psi_l}{\sqrt{N!}}  \psi\left(\sum_i x_i-cN\tau+\frac{NL}{2}\right)\left(\frac{\sigma^2}{\pi}\right)^{N/4}\exp\left[-\frac{\sigma^2}{2}\sum_{i=1}^N(x_i-c\tau+\frac{L}{2})^2\right]\prod_{i=1}^Na_{x_i,l}^\dagger\ket{0}\\
&\rightarrow \frac{\psi_l}{\sqrt{N!}}  \psi\left(\sum_i x_i-cN\tau+\frac{NL}{2}\right)\left(\frac{\sigma^2}{\pi}\right)^{N/4}\exp\left[-\frac{\sigma^2}{2}\sum_{i=1}^N(x_i-c\tau+\frac{L}{2})^2\right]\prod_{i=1}^Na_{x_i+c(t-\tau),l}^\dagger\ket{0}.
\end{aligned}
\end{equation}
Assume that at time $\tau$, all positions $x_k < y$, so each mode will reach Charlie at position $y$ at time $T_k = \tau + (y - x_k)/c$. If  the $k$-th mode is the first to arrive at $y$
\begin{equation}
\begin{aligned}
\rightarrow\frac{\psi_l}{\sqrt{N!}}  \psi\left(\sum_i x_i-cN\tau+\frac{NL}{2}\right)\left(\frac{\sigma^2}{\pi}\right)^{N/4}\exp\left[-\frac{\sigma^2}{2}\sum_{i=1}^N(x_i-c\tau+\frac{L}{2})^2\right]a_{y,l}^\dagger\prod_{i\neq k}a_{x_i+c(T_k-\tau),l}^\dagger\ket{0}.
\end{aligned}
\end{equation}
Charlie applies the beam splitter operation, which transforms the creation operator as
$a_{y,l}^\dagger \rightarrow U_{00} a_{y,l}^\dagger + U_{10} a_{y,r}^\dagger$.
\begin{equation}
\begin{aligned}
\rightarrow\frac{\psi_l}{\sqrt{N!}}  \psi\left(\sum_i x_i-cN\tau+\frac{NL}{2}\right)\left(\frac{\sigma^2}{\pi}\right)^{N/4}\exp\left[-\frac{\sigma^2}{2}\sum_{i=1}^N(x_i-c\tau+\frac{L}{2})^2\right](U_{00}a_{y,l}^\dagger+U_{10}a_{y,r}^\dagger)\prod_{i\neq k}a_{x_i+c(T_k-\tau),l}^\dagger\ket{0}.
\end{aligned}
\end{equation}
The mode then continues to propagate as
\begin{equation}
\begin{aligned}
\rightarrow\frac{\psi_l}{\sqrt{N!}}  \psi\left(\sum_i x_i-cN\tau+\frac{NL}{2}\right)\left(\frac{\sigma^2}{\pi}\right)^{N/4}\exp\left[-\frac{\sigma^2}{2}\sum_{i=1}^N(x_i-c\tau+\frac{L}{2})^2\right](U_{00}a_{y+c(t-T_k),l}^\dagger+U_{10}a_{y-c(t-T_k),r}^\dagger)\prod_{i\neq k}a_{x_i+c(t-\tau),l}^\dagger\ket{0}.
\end{aligned}
\end{equation}
Similarly, beam splitter operations are applied to all modes in the same manner.
\begin{equation}
\begin{aligned}
\rightarrow\frac{\psi_l}{\sqrt{N!}}  \psi\left(\sum_i x_i-cN\tau+\frac{NL}{2}\right)\left(\frac{\sigma^2}{\pi}\right)^{N/4}\exp\left[-\frac{\sigma^2}{2}\sum_{i=1}^N(x_i-c\tau+\frac{L}{2})^2\right]\prod_{k=1}^N(U_{00}a_{y+c(t-T_k),l}^\dagger+U_{10}a_{y-c(t-T_k),r}^\dagger)\ket{0}.
\end{aligned}
\end{equation}
Similarly, the evolution of the terms originating from Bob's side can be determined in the same way
\begin{equation}
\begin{aligned}
&\frac{\psi_r}{\sqrt{N!}}\int d\vec{x}     \psi\left(\sum_i x_i+cN\tau-\frac{NL}{2}\right)\left(\frac{\sigma^2}{\pi}\right)^{N/4}\exp\left[-\frac{\sigma^2}{2}\sum_{i=1}^N(x_i+c\tau-\frac{L}{2})^2\right]\prod_{i=1}^Na_{x_i,r}^\dagger\ket{0}\\
&\rightarrow\frac{\psi_r}{\sqrt{N!}}\int d\vec{x}     \psi\left(\sum_i x_i+cN\tau-\frac{NL}{2}\right)\left(\frac{\sigma^2}{\pi}\right)^{N/4}\exp\left[-\frac{\sigma^2}{2}\sum_{i=1}^N(x_i+c\tau-\frac{L}{2})^2\right]\prod_{k=1}^N(U_{01}a_{y+c(t-T_k'),l}^\dagger+U_{11}a_{y-c(t-T_k'),r}^\dagger)\ket{0},
\end{aligned}
\end{equation}
where $T_k'=\tau+(x_k-y)/c$. After substituting the expressions for $T_k$ and $T_k'$ and redefining the coordinates $x_i$ by shifting them by a constant, we obtain the resulting states after Charlie's beam splitter operations
\begin{equation}
\begin{aligned}\label{eq:phi_SI}
\ket{\phi}&=\frac{\psi_l}{\sqrt{N!}}\int d\vec{x}     \psi\left(\sum_i x_i-cNt+\frac{NL}{2}\right)\left(\frac{\sigma^2}{\pi}\right)^{N/4} \exp\left[-\frac{\sigma^2}{2}\sum_{i=1}^N(x_i-ct+\frac{L}{2})^2\right]\prod_{i=1}^N(U_{00}a_{x_i,l}^\dagger+U_{10}a_{2y-x_i,r}^\dagger)\ket{0}\\
&+\frac{\psi_r}{\sqrt{N!}}\int d\vec{x}     \psi\left(\sum_i x_i+cNt-\frac{NL}{2}\right)\left(\frac{\sigma^2}{\pi}\right)^{N/4}\exp\left[-\frac{\sigma^2}{2}\sum_{i=1}^N(x_i+ct-\frac{L}{2})^2\right]\prod_{i=1}^N(U_{01}a_{2y-x_i,l}^\dagger+U_{11}a_{x_i,r}^\dagger)\ket{0}.\\
\end{aligned}
\end{equation}

We now calculate the quantum Fisher information (QFI) for estimating Charlie’s position $y$ based on the relation between QFI and fidelity, $F=\frac{8}{\delta y^2}[1-|\bra{\phi_y}\ket{\phi_{y+\delta y}}|]$ \cite{braunstein1994statistical,paris2009quantum,kay1993fundamentals,sidhu2020geometric}
\begin{equation}
\begin{aligned}
\bra{\phi_y}\ket{\phi_{y+\delta y}}=\sum_{\vec{i}}&\bigg[|\psi_l|^2(\prod_{j=1}^N|U_{i_j,0}|^2)\delta_{11}(\vec{i},\delta y)+\psi_l^*\psi_r(\prod_{j=1}^N U_{i_j,0}^*U_{i_j,1})\delta_{12}(\vec{i},\delta y)\\
&+\psi_l\psi_r^*(\prod_{j=1}^N U_{i_j,1}^*U_{i_j,0})\delta_{21}(\vec{i},\delta y)+|\psi_r|^2(\prod_{j=1}^N|U_{i_j,1}|^2)\delta_{22}(\vec{i},\delta y)\bigg],\\
\end{aligned}
\end{equation}
\begin{equation}
\begin{aligned}
&\delta_{pq}(\vec{i},\delta y)=\int d\vec{z}f_p^*(y) f_q(y+\delta y),\\
&f_1(y)=\psi\left(Ny+\sum_{j=1}^N(-1)^{i_j}(z_j-y)+\frac{NL}{2}-cNt\right)\left(\frac{\sigma^2}{\pi}\right)^{N/4}\exp(-\frac{\sigma^2}{2}\sum_{j=1}^N\left((-1)^{i_j}(z_j-y)+y+\frac{L}{2}-ct\right)^2),\\
&f_2(y)=\psi\left(Ny-\sum_{j=1}^N(-1)^{i_j}(z_j-y)-\frac{NL}{2}+cNt\right)\left(\frac{\sigma^2}{\pi}\right)^{N/4}\exp(-\frac{\sigma^2}{2}\sum_{j=1}^N\left(-(-1)^{i_j}(z_j-y)+y-\frac{L}{2}+ct\right)^2),\\
\end{aligned}
\end{equation}
where we set $l = 0$ and $r = 1$, and will use $l, r$ and $0, 1$ interchangeably in the following, $\sum_{\vec{i}}=\sum_{i_1=l,r}\sum_{i_2=l,r}\cdots\sum_{i_N=l,r}$. After performing the integration and taking the limit $\sigma \to 0$, we obtain
\begin{equation}
\begin{aligned}
&\delta_{11}=\int dr \sqrt{\frac{\sigma^2}{N\pi}}\psi^*(r)\psi(r+2Q\delta y),\\
&\delta_{12}=\int dr \sqrt{\frac{\sigma^2}{N\pi}}\psi^*(r)\psi(-r+2Ny+(2N-2Q)\delta y),\\
&\delta_{21}=\int dr \sqrt{\frac{\sigma^2}{N\pi}}\psi^*(r)\psi(-r+2Ny+2Q\delta y),\\
&\delta_{22}=\int dr \sqrt{\frac{\sigma^2}{N\pi}}\psi^*(r)\psi(r+(2N-2Q)\delta y),\\
\end{aligned}
\end{equation}
where $Q = \sum_{k=1}^N \delta_{i_k,1}$ denotes the number of indices $k$ such that $i_k = 1$.
If we choose 
\begin{equation}\label{eq:psi_k}
\tilde{\psi}(k)=\frac{N^{1/4}}{(4\pi \sigma^2)^{1/4}}\exp(-k^2/2\beta^2)/(\pi\beta^2)^{1/4} .
\end{equation}
\begin{equation}
\begin{aligned}
&\delta_{11}(Q,\delta y)=\exp(-\beta^2Q^2\delta y^2),\\
&\delta_{12}(Q,\delta y)=\exp(-\beta^2[\delta y(N-Q)+N y]^2),\\
&\delta_{21}(Q,\delta y)=\exp(-\beta^2[\delta yQ+N y]^2),\\
&\delta_{22}(Q,\delta y)=\exp(-\beta^2(N-Q)^2\delta y^2).\\
\end{aligned}
\end{equation}

We then have the QFI
\begin{equation}
\begin{aligned}
&F=\frac{8}{\delta y^2}[1-|\bra{\phi_y}\ket{\phi_{y+\delta y}}|]\\
&=-16\exp(-2\beta^2 N^2 y^2)N^2y^2\beta^4\left|\sum_{\vec{i}}\psi_l^*\psi_r (\prod_{j=1}^N U_{i_j,0}^* U_{i_j,1})(N-Q)+\psi_l\psi^*_r (\prod_{j=1}^N U_{i_j,1}^* U_{i_j,0})Q\right|^2\\
&+4\sum_{\vec{i}}\bigg[|\psi_l|^2(\prod_{j=1}^N |U_{i_j,0}|^2) 2\beta^2 Q^2+\psi_l^*\psi_r(\prod_{j=1}^N U_{i_j,0}^* U_{i_j,1})\exp(-\beta^2 N^2 y^2)\beta^2((N-Q)^2+Q^2)(1-2\beta^2 N^2 y^2)\\
&\quad\quad\quad+\psi_l\psi^*_r(\prod_{j=1}^N U_{i_j,1}^* U_{i_j,0})\exp(-\beta^2 N^2 y^2)\beta^2((N-Q)^2+Q^2)(1-2\beta^2 N^2 y^2)+|\psi_r|^2(\prod_{j=1}^N |U_{i_j,1}|^2) 2\beta^2 (N-Q)^2\bigg].
\end{aligned}
\end{equation}

If we choose 
\begin{equation}
\begin{aligned}
U=\frac{1}{\sqrt{2}}\left[\begin{matrix}
1 & 1\\
-1 & 1
\end{matrix}\right],
\end{aligned}
\end{equation}
\begin{equation}
\begin{aligned}
&F=-16\exp(-2\beta^2 N^2 y^2)N^2y^2\beta^4\left|\sum_{Q=0}^N C_N^Q[\psi_l^*\psi_r \frac{(-1)^Q}{2^N}(N-Q)+\psi_l\psi^*_r \frac{(-1)^Q}{2^N}Q]\right|^2\\
&+4\sum_{Q=0}^N\bigg[|\psi_l|^2\frac{1}{2^N} 2\beta^2 Q^2+\psi_l^*\psi_r\frac{(-1)^Q}{2^N}\exp(-\beta^2 N^2 y^2)\beta^2((N-Q)^2+Q^2)(1-2\beta^2 N^2 y^2)\\
&\quad\quad\quad+\psi_l\psi^*_r\frac{(-1)^Q}{2^N}\exp(-\beta^2 N^2 y^2)\beta^2((N-Q)^2+Q^2)(1-2\beta^2 N^2 y^2)+|\psi_r|^2\frac{1}{2^N} 2\beta^2 (N-Q)^2\bigg]\\
&=2\beta^2N(N+1),
\end{aligned}
\end{equation}
where we use the fact that $\sum_{Q=0}^N Q^2 C_N^Q=N(N+1)2^{N-2}$, $\sum_{Q=0}^N (-1)^Q Q C_N^Q=\sum_{Q=0}^N (-1)^Q Q^2 C_N^Q=0$.

%If $y=0$, 
%\begin{equation}
%\begin{aligned}
%&F=\frac{8}{\delta y^2}[1-|\bra{\phi_y}\ket{\phi_{y+\delta y}}|]\\
%&=4\sum_{i_1=0,1}\cdots\sum_{i_N=0,1}\bigg[|\psi_l|^2(\prod_{j=1}^N |U_{i_j,0}|^2) 2\beta^2 Q^2+\psi_l^*\psi_r(\prod_{j=1}^N U_{i_j,0}^* U_{i_j,1})\beta^2((N-Q)^2+Q^2)\\
%&\quad\quad\quad+\psi_l\psi^*_r(\prod_{j=1}^N U_{i_j,1}^* U_{i_j,0})\beta^2((N-Q)^2+Q^2)+|\psi_r|^2(\prod_{j=1}^N |U_{i_j,1}|^2) 2\beta^2 (N-Q)^2\bigg]
%\end{aligned}
%\end{equation}

\section{Security of the positioning}\label{SI:security}

\subsection{Derivation of the states prepared by cheaters}\label{SI:state_cheater}

In this subsection, we derive the states that can be prepared by the cheaters.
Alice and Bob still send the initial states as given in Eq.~\ref{eq:SM_psi}.
Assume that the two cheaters are positioned at $y_{1,2} = y' \mp d/2$, where $y' \neq y$ is the location at which they aim to impersonate a fake prover Charlie. Each cheater is allowed to perform local two-port beam splitter operations, meaning that an ancillary mode can be introduced at each location. In addition, they are assumed to have access to quantum memory, allowing them to delay the pulses arbitrarily without altering their shape.

We first consider the terms in $\ket{\psi}$ corresponding to the modes on the left, originating from Alice.
Assume that at time $\tau$, all positions $x_k < y' - d/2$ for the modes $a_{x_k}$. Then, each mode will reach the left cheater $C_l$ at time
\begin{equation}
T_k=\tau+\frac{y'-\frac{d}{2}-x_k}{c}.
\end{equation}
If the $k$-th mode reaches the left cheater's position $y_1 = y' - d/2$ first, then the left cheater $C_l$ applies a beam splitter operation
\begin{equation}
V=\left[\begin{matrix}
V_{00} & V_{01}\\
V_{10} & V_{11}
\end{matrix}\right].
\end{equation}
The terms on the left, originating from Alice, evolve as follows
\begin{equation}
\begin{aligned}
&\frac{\psi_l}{\sqrt{N!}}     \psi\left(\sum_i x_i-cN\tau+\frac{NL}{2}\right)\left(\frac{\sigma^2}{\pi}\right)^{N/4} \exp\left[-\frac{\sigma^2}{2}\sum_{i=1}^N(x_i-c\tau+\frac{L}{2})^2\right]\prod_{i=1}^Na_{x_i,l}^\dagger\ket{0}\\
&\rightarrow\frac{\psi_l}{\sqrt{N!}}     \psi\left(\sum_i x_i-cN\tau+\frac{NL}{2}\right)\left(\frac{\sigma^2}{\pi}\right)^{N/4} \exp\left[-\frac{\sigma^2}{2}\sum_{i=1}^N(x_i-c\tau+\frac{L}{2})^2\right]a_{y_1,l}^\dagger\prod_{i\neq k}a_{x_i+c(T_k-\tau),l}^\dagger\ket{0}\\
&\rightarrow\frac{\psi_l}{\sqrt{N!}}     \psi\left(\sum_i x_i-cN\tau+\frac{NL}{2}\right)\left(\frac{\sigma^2}{\pi}\right)^{N/4} \exp\left[-\frac{\sigma^2}{2}\sum_{i=1}^N(x_i-c\tau+\frac{L}{2})^2\right](V_{00}a_{y_1,l}^\dagger+V_{01}b_{y_1,l}^\dagger)\prod_{i\neq k}a_{x_i+c(T_k-\tau),l}^\dagger\ket{0},\\
\end{aligned}
\end{equation}
where $b_{y_1,l}$ denotes the ancillary mode introduced at the left cheater $C_l$. The cheater stores the $b_{y_1,l}$ mode in a quantum memory, while allowing the $a_{y_1,l}$ mode to continue propagating
\begin{equation}
\begin{aligned}
&\rightarrow\frac{\psi_l}{\sqrt{N!}}     \psi\left(\sum_i x_i-cN\tau+\frac{NL}{2}\right)\left(\frac{\sigma^2}{\pi}\right)^{N/4} \exp\left[-\frac{\sigma^2}{2}\sum_{i=1}^N(x_i-c\tau+\frac{L}{2})^2\right](V_{00}a_{y_1+c(t-T_k),l}^\dagger+V_{01}b_{y_1,l}^\dagger)\prod_{i\neq k}a_{x_i+c(t-\tau),l}^\dagger\ket{0}.\\
\end{aligned}
\end{equation}
Similarly, each mode passes through a beam splitter, and the left cheater $C_l$ always retains the $b_{y_1,l}$ mode in quantum memory. This results in the state
\begin{equation}
\begin{aligned}
&\rightarrow\frac{\psi_l}{\sqrt{N!}}     \psi\left(\sum_i x_i-cN\tau+\frac{NL}{2}\right)\left(\frac{\sigma^2}{\pi}\right)^{N/4} \exp\left[-\frac{\sigma^2}{2}\sum_{i=1}^N(x_i-c\tau+\frac{L}{2})^2\right]\prod_{k=1}^N(V_{00}a_{y_1+c(t-T_k),l}^\dagger+V_{01}b_{y_1,l}^\dagger)\ket{0},\\
\end{aligned}
\end{equation}
where we assume $t-T_k\leq d/c$.

Similarly, if the right cheater $C_r$ performs any two-port beam splitter operation with an ancillary mode
\begin{equation}
W=\left[\begin{matrix}
W_{00} & W_{01}\\
W_{10} & W_{11}
\end{matrix}\right],
\end{equation}
we can similarly derive the states for the terms on the right originating from Bob, which yields
\begin{equation}
\begin{aligned}
&\frac{\psi_r}{\sqrt{N!}}     \psi\left(\sum_i x_i+cN\tau-\frac{NL}{2}\right)\left(\frac{\sigma^2}{\pi}\right)^{N/4} \exp\left[-\frac{\sigma^2}{2}\sum_{i=1}^N(x_i+c\tau-\frac{L}{2})^2\right]\prod_{k=1}^N(W_{00}a_{y_2-c(t-T_k'),r}^\dagger+W_{01}b_{y_2,r}^\dagger)\ket{0},\\
\end{aligned}
\end{equation}
where $y_2=y'+d/2$,
\begin{equation}
T_k'=\tau+\frac{x_k-y'-\frac{d}{2}}{c} ,
\end{equation}
and we assume $t-T_k'\leq d/c$.

After the first round of operations by the two cheaters, they each retain the modes $b_{y_1,l}$ and $b_{y_2,r}$ locally for a duration of $d/c$, in order to align with the original pulse shape. At time $t$, if there is no further operation, the state evolves to
\begin{equation}
\begin{aligned}
\ket{\gamma'}&=\frac{\psi_l}{\sqrt{N!}}\int d\vec{x}     \psi\left(\sum_i x_i-cN\tau+\frac{NL}{2}\right)\left(\frac{\sigma^2}{\pi}\right)^{N/4} \exp\left[-\frac{\sigma^2}{2}\sum_{i=1}^N(x_i-c\tau+\frac{L}{2})^2\right]\\
&\quad\times\prod_{k=1}^N(V_{00}a_{x_k+c(t-\tau),l}^\dagger+V_{01}b_{2y'-x_k-c(t-\tau),l}^\dagger)\ket{0}\\
&+\frac{\psi_r}{\sqrt{N!}}\int d\vec{x}     \psi\left(2Ny-\sum_i x_i+cN\tau-\frac{NL}{2}\right)\left(\frac{\sigma^2}{\pi}\right)^{N/4} \exp\left[-\frac{\sigma^2}{2}\sum_{i=1}^N(2y'-x_i+c\tau-\frac{L}{2})^2\right]\\
&\quad\times\prod_{k=1}^N(W_{00}a_{2y'-x_k+c(\tau-t),r}^\dagger+W_{01}b_{x_k-c(\tau-t),r}^\dagger)\ket{0}.\\
\end{aligned}
\end{equation}
But the cheaters will implement another beam splitter at each side.
At this point, the modes $a_{z,l}$ and $b_{z,r}$ are located at the right cheater $C_r$, while the modes $b_{z,l}$ and $a_{z,r}$ are at the left cheater $C_l$. Assume the right cheater now performs a local beam splitter operation
\begin{equation}\begin{aligned}
&a_{x_k+c(t-\tau),l}^\dagger\rightarrow P_{00}a_{x_k+c(t-\tau),l}^\dagger+P_{01}b_{x_k+c(t-\tau),r}^\dagger\\
&b_{x_k+c(t-\tau),r}^\dagger\rightarrow P_{10}a_{x_k+c(t-\tau),l}^\dagger+P_{11}b_{x_k+c(t-\tau),r}^\dagger\\
\end{aligned}
\end{equation}
Similarly, the left cheater $C_l$ does a local beam splitter operation
\begin{equation}\begin{aligned}
&a_{2y'-x_k+c(\tau-t),r}^\dagger\rightarrow Q_{00}a_{2y'-x_k+c(\tau-t),r}^\dagger+Q_{01}b_{2y'-x_k+c(\tau-t),l}^\dagger\,\\
&b_{2y'-x_k+c(\tau-t),l}^\dagger\rightarrow Q_{10}a_{2y'-x_k+c(\tau-t),r}^\dagger+Q_{11}b_{2y'-x_k+c(\tau-t),l}^\dagger\,\\
\end{aligned}
\end{equation}
We then get the states prepared by the cheaters
\begin{equation}
\begin{aligned}
\ket{\gamma}&=\frac{\psi_l}{\sqrt{N!}}\int d\vec{x}     \psi\left(\sum_i x_i-cNt+\frac{NL}{2}\right)\left(\frac{\sigma^2}{\pi}\right)^{N/4} \exp\left[-\frac{\sigma^2}{2}\sum_{i=1}^N(x_i-ct+\frac{L}{2})^2\right]\\
&\quad\times\prod_{k=1}^N(V_{00}P_{00}a_{x_k,l}^\dagger+V_{00}P_{01}b_{x_k,r}^\dagger+V_{01}Q_{10}a_{2y'-x_k,r}^\dagger+V_{01}Q_{11}b_{2y'-x_k,l}^\dagger)\ket{0}\\
&+\frac{\psi_r}{\sqrt{N!}}\int d\vec{x}     \psi\left(2Ny-\sum_i x_i+cNt-\frac{NL}{2}\right)\left(\frac{\sigma^2}{\pi}\right)^{N/4} \exp\left[-\frac{\sigma^2}{2}\sum_{i=1}^N(2y-x_i+ct-\frac{L}{2})^2\right]\\
&\quad\times\prod_{k=1}^N(W_{01}P_{10}a_{x_k,l}^\dagger+W_{01}P_{11}b_{x_k,r}^\dagger+W_{00}Q_{00}a_{2y'-x_k,r}^\dagger+W_{00}Q_{01}b_{2y'-x_k,l}^\dagger)\ket{0}.\\
\end{aligned}
\end{equation}

When Alice and Bob measure the state $\ket{\gamma_{y'}}$, they obtain an estimated position $y''$, which the cheaters aim to match with their chosen fake location $y' \neq y$. For generality, we simply assume $y'' \neq y$, which may or may not equal $y'$. To compare this with the state prepared by a genuine operation from Charlie, as given in Eq.~\ref{eq:phi}, we compute the inner product between the two states as
\begin{equation}\label{SI_eq:phi_gamma}
\begin{aligned}
\bra{\phi_{y''}}\ket{\gamma_{y'}}=\sum_{\vec{i}}&\bigg[|\psi_l|^2(\prod_{j=1}^NU_{i_j,0}^*U_{i_j,0}')\delta_{11}+\psi_l^*\psi_r(\prod_{j=1}^N U_{i_j,0}^*U_{i_j,1}')\delta_{12}\\
&+\psi_l\psi_r^*(\prod_{j=1}^N U_{i_j,1}^*U_{i_j,0}')\delta_{21}+|\psi_r|^2(\prod_{j=1}^N U_{i_j,1}^*U_{i_j,1}')\delta_{22}\bigg],\\
\end{aligned}
\end{equation}
\begin{equation}
\begin{aligned}
&\delta_{pq}(\vec{i},\delta y)=\int d\vec{z}f_p^*(y'')f_q(y'),\\
&f_1(y)=\psi\left(Ny+\sum_{j=1}^N(-1)^{i_j}(z_j-y)+\frac{NL}{2}-cNt\right)\left(\frac{\sigma^2}{\pi}\right)^{N/4}\exp(-\frac{\sigma^2}{2}\sum_{j=1}^N\left((-1)^{i_j}(z_j-y)+y+\frac{L}{2}-ct\right)^2),\\
&f_2(y)=\psi\left(Ny-\sum_{j=1}^N(-1)^{i_j}(z_j-y)-\frac{NL}{2}+cNt\right)\left(\frac{\sigma^2}{\pi}\right)^{N/4}\exp(-\frac{\sigma^2}{2}\sum_{j=1}^N\left(-(-1)^{i_j}(z_j-y)+y-\frac{L}{2}+ct\right)^2),\\
\end{aligned}
\end{equation}
where $U_{00}'=V_{00}P_{00}$, $U_{10}'=V_{01}Q_{10}$, $U_{01}'=W_{01}P_{10}$, $U_{11}'=W_{00}Q_{00}$.  Using the choice of $\tilde{\psi}(k)=\frac{N^{1/4}}{(4\pi \sigma^2)^{1/4}}\exp(-k^2/2\beta^2)/(\pi\beta^2)^{1/4} $, we can find
\begin{equation}
\begin{aligned}
&\delta_{11}=\exp(-\beta^2Q^2\Delta y^2),\quad \delta_{12}=\exp(-\beta^2[\Delta y(N-Q)+Ny'']^2),\\
&\delta_{21}=\exp(-\beta^2[\Delta yQ+Ny'']^2),\quad \delta_{22}=\exp(-\beta^2(N-Q)^2\Delta y^2),
\end{aligned}
\end{equation}
where $\Delta y=y'-y''$.

\subsection{Proof of Proposition~\ref{prop:reflection} and \ref{prop:single_side}}\label{proof_prop:reflection}
We begin by proving Proposition~\ref{prop:reflection}. Given that $U_{01} = U_{10} = 1$ and $U_{00} = U_{11} = 0$, we can choose the parameters as $V_{01} = Q_{10} = W_{01} = P_{10} = 1$ and $V_{00} = P_{00} = W_{00} = Q_{00} = 0$. This choice reproduces $U_{pq} = U_{pq}'$ exactly. As a result,   the cheaters have successfully prepared the state $\ket{\gamma_{y'}} = \ket{\phi_{y'}}$. Since $\ket{\gamma_{y'}}$ is exactly the same as $\ket{\phi_{y'}}$, Alice and Bob will always obtain $y'' = y'$ in this case.

We now prove Proposition~\ref{prop:single_side}. Suppose the cheaters only prepare the state from one side, with fixed inputs $\psi_l = 1$ and $\psi_r = 0$. In this case, the cheaters can trivially reproduce the same state that Charlie would have prepared, after discarding the original state, thereby spoofing his position.
Moreover, even if the cheaters do not actively prepare the initial states, they can still perform the same operations as Charlie. Specifically, we can always choose the parameters $P_{00} = Q_{10} = 1$, $V_{00} = U_{00}$, and $V_{01} = U_{10}$ to achieve $\ket{\gamma_{y'}}=\ket{\phi_{y'}}$ (Alice and Bob also always obtain $y''=y'$ in this case). This configuration corresponds to the scenario in which the light is emitted solely from Alice’s side.
In such a case, the left-side cheater can faithfully replicate Charlie’s operations. Furthermore, the cheater can store the light locally in a quantum memory to synchronize its output with the expected timing of Charlie’s operations. In this way, the cheaters can prepare a state indistinguishable from the honest one.
This example highlights the necessity of sending light from both Alice and Bob’s sides in order to maintain the security of the positioning protocol.

\subsection{Proof of Theorem \ref{thm:optimal_gamma}}\label{proof_thm:optimal_gamma}

%\am{I'm a bit concerned about this proof: In particular the strategy for detecting cheating in the case where the cheaters use a non-local computation, and the case where they prepare a fixed state, are different, but Alice and Bob don't know which strategy the adversary is going to follow. It seems like Alice and Bob should use just one strategy to detect cheating.}\yunkai{I agree. Thanks for pointing this out. I think the measurement strategy suggested by you yesterday should also work. I will try to find a lower bound in this way.}

To upper-bound the error probability $P_1$, corresponding to the case where the cheaters perform operations on the original states, we make the following assumptions. First, we assume that the cheaters choose the optimal operations. Second, we assume the Alice and Bob perform a specific (possibly suboptimal) measurement.
On the detection side, Alice and Bob use a POVM of the form $\{M_0 = \ket{\phi_i}\bra{\phi_i}, M_1 = I - \ket{\phi_i}\bra{\phi_i}\}$ when they send the state $\ket{\phi_i}$. Note that the state $\ket{\phi_i}$ depends on their estimated location $y''$, which must be determined before any detection of potential cheating can occur. We assume that Alice and Bob have access to a sufficiently large number of samples to estimate $y''$, and that in the absence of cheating, $y'' = y$.
The decision strategy is as follows: Alice and Bob conclude that no cheater is present upon observing outcome $M_0$, and conclude that cheating has occurred upon observing outcome $M_1$. Clearly, if no cheaters are present, the outcome will always be $M_0$, and thus no false alarms will occur. 
\begin{equation}
P_1(\text{cheater}|\text{honest})=0.
\end{equation}
The error probability of mistakenly accepting a cheated state as honest—i.e., the probability of outcome $M_0$ when the cheaters are present—is given by
\begin{align}
P_1(\text{honest}|\text{cheater})= \sum_i p_i |\bra{\phi_{iy''}}\ket{\gamma_{iy'}}|^2 ,
\end{align}
where, to keep the analysis general, we allow $y''$—the position estimated by Alice and Bob—to differ from the position $y'$ chosen by the cheaters.
Observe that the vectors $u_0 = [U_{00}, U_{10}]^T$ and $u_1 = [U_{01}, U_{11}]^T$ are orthonormal, as $U$ is a unitary matrix. The same applies to the vectors constructed from $V_{ij}, W_{ij}, P_{ij}$, and $Q_{ij}$ individually. However, the modified vector $u_0' = [U_{00}', U_{10}']^T$ incorporates components from both $v_0 = [V_{00}, V_{01}]^T$ and the additional coefficients $P_{00}$ and $Q_{10}$. As a result, it is generally not possible to simultaneously satisfy $u_0' = u_0$ and $u_1' = u_1$.
The quantities $\sum_{i=0,1} U_{i,p}^* U_{i,q}' = u_p^\dagger u_q'$ represent inner products between two vectors, where one of the vectors—$u_q'$—typically has norm less than one. As long as all inner products $u_p^\dagger u_q'$ for $p, q = 0, 1$ are strictly less than one, increasing $N$ sufficiently will cause the product $\prod_{j=1}^N u_{p_j}^\dagger u_{q_j}'$ to decay exponentially toward zero. 
Nevertheless, if the cheaters configure their operations so that one inner product equals one—e.g., by choosing $P_{00} = Q_{10} = 1$, $P_{10} = Q_{00} = 0$, and setting $V_{00} = U_{00}$, $V_{01} = U_{10}$—they can always approximate half of the operations perfectly. This implies 
\begin{equation}\label{eq:U'_half}
U'=\left[\begin{matrix}
U_{00} & 0\\
U_{10} & 0
\end{matrix}\right].
\end{equation}
More rigorous proof can be given by observing the factor from Eq. \ref{SI_eq:phi_gamma} is bounded above by
\begin{equation}\begin{aligned}\label{SI_eq:U_U_delta}
&\sum_{\vec{i}}\prod_{j=1}^NU_{i_j,p}^*U_{i_j,q}'\delta_{p+1,q+1}=\sum_{Q=0}^NC_N^Qa^{N-Q}b^Q\delta_{p+1,q+1} \\
&\leq\sum_{Q=0}^NC_N^Q|a|^{N-Q}|b|^Q\delta_{p+1,q+1}\leq \sum_{Q=0}^NC_N^Q|a|^{N-Q}|b|^Q=(|a|+|b|)^N\leq|u_q'|^{N},\\
\end{aligned}
\end{equation}
\begin{equation}\begin{aligned}\label{SI_eq:ab_u}
&(|a|+|b|)^2=(|U_{0p}||U'_{0q}|+|U_{1p}||U'_{1q}|)^2\leq(|U_{0p}|^2+|U_{1p}|^2)(|U'_{0q}|^2+|U'_{1q}|^2)=|u_q'|^2,
\end{aligned}
\end{equation}
where $a=U^*_{0p}U'_{0q}$, $b=U^*_{1p}U'_{1q}$, where in Eq. \ref{SI_eq:ab_u}, we used the Cauchy–Schwarz inequality, and the equality holds when $u_p$ is parallel to $u_q'$. We can make $|u_0'|=|u_1'|$ only when $U'$ takes the following two form.
\begin{equation}
U'=\left[\begin{matrix}
1 & 0\\
0 & e^{i\theta}
\end{matrix}\right]\quad \text{or}\quad \left[\begin{matrix}
0 & e^{i\alpha}\\
e^{i\beta} & 0
\end{matrix}\right].
\end{equation}
But in these case, we cannot have $u_p$ parallel to $u_q'$ if $U$ is chosen as a general unitary matrix, which means equality cannot be hold. Whenever, $|u_q'|<1$, we have $|u_q'|^N\rightarrow0$. As motivated above, we can make one of the $|u_q'|=1$ by approximating half of the operations perfectly. 
And in this case, we have
\begin{equation}\label{SI_eq:product_phi_gamma}
|\bra{\phi_{iy''}}\ket{\gamma_{iy'}}|\leq \max\{|\psi_{il}|^2 , |\psi_{ir}|^2 \},
\end{equation}
which is independent of $N$.  So, the error probability is upper bounded by
\begin{equation}
P_1(\text{honest}|\text{cheater})\leq \sum_ip_i\max\{|\psi_{il}|^2,|\psi_{ir}|^2\}.
\end{equation}
And the total error probability for the case in which the cheaters directly act on the original state is given by
\begin{equation}
P_1=P_1(\text{honest}|\text{cheater})+P_1(\text{cheater}|\text{honest})\leq \sum_ip_i\max\{|\psi_{il}|^2,|\psi_{ir}|^2\},
\end{equation}
which completes the proof.

We now derive a bound on the total error probability $P_2$ in the case where the cheaters discard the original state and instead prepare their own state $\rho_0$.
For the lower bound, consider the scenario where the cheaters prepare the state $\rho_0 = \sum_i p_i \rho_i$
\begin{equation}
\begin{aligned}
P_2&=\sum_i p_i \left(1-\frac{1}{2}\|\rho_0-\rho_i\|\right)=1-\sum_i p_i \left(\frac{1}{2}\|\sum_jp_j\rho_j-\rho_i\|\right)\geq1-\frac{1}{2}\sum_{i,j}p_i p_j\|\rho_i-\rho_j\|\\
&\geq 1-\sum_{i\neq j}p_ip_j.
\end{aligned}
\end{equation}
We can also upper bound the error probability. In general, Alice and Bob do not know which specific state $\rho_0$ is prepared by the cheaters. However, once they obtain an estimate of the position $y''$, they still perform a measurement using the POVM $\{M_0 = \ket{\phi_{iy''}}\bra{\phi_{iy''}}, M_1 = I - \ket{\phi_{iy''}}\bra{\phi_{iy''}}\}$ when using the state $\ket{\psi_i}$. Their decision rule is to conclude that no cheaters are present if the outcome is $M_0$, and to conclude the presence of cheaters if the outcome is $M_1$. Under this strategy, the probability that the cheaters pass the test is given by
\begin{align}
   P_2(\text{honest}|\text{cheater}) = \sum_i p_i |\bra{\phi_{iy''}}\rho_0\ket{\phi_{iy''}}| = \sum_i p_i F(\rho_0, \phi_{iy''}).
\end{align}
Ref.~\cite{afham2022quantum} provides the following upper bound:
$ \max_{\rho_0} \sum_i p_i \sqrt{F}(\rho_0, \phi_{iy''}) \leq \sqrt{ \sum_{i,j} p_i p_j \sqrt{F(\phi_{iy''}, \phi_{jy''})} } = \sqrt{ \sum_{i,j} p_i p_j \sqrt{F(\psi_i, \psi_j) }}$.
Using the inequality $F \leq \sqrt{F}$, we can get
\begin{align}
    P_2(\text{honest}|\text{cheater}) \leq \sqrt{\sum_{i,j}p_i p_j \sqrt{F(\psi_{i}, \psi_{j})}}.
\end{align}
Since, in the absence of cheaters, Alice and Bob will always obtain $y'' = y$ and observe the outcome $M_0$, the conditional success probability in this case is $100\%$. Therefore, the conditional error probability when the prover is honest is
\begin{equation} P_2(\text{cheater}|\text{honest}) = 0. \end{equation}
Hence, the total error probability $P_2$ is bounded from above by
\begin{equation}
P_2=P_2(\text{honest}|\text{cheater}) +P_2(\text{cheater}|\text{honest}) \leq \sqrt{\sum_{i,j}p_i p_j |\bra{\psi_i}\ket{ \psi_j}|}.
\end{equation}

\section{Simultaneous estimation of the position and the detection of the cheaters}\label{SI:simultaneous}
\subsection{Fisher information calculation}
To simultaneously determine the position and evaluate the presence of cheaters, we consider a specific measurement strategy that achieves Heisenberg scaling in the Fisher information (FI), while ensuring that the error probability vanishes as the number of samples increases. We consider a measurement in which Alice and Bob interfere the light they receive using a beam splitter that implements the following operations
\begin{equation}
\begin{aligned}
a_{z,0}^\dagger=R_{00}b_{z,0}^\dagger+R_{01}b_{-z,1}^\dagger,\quad a_{-z,1}^\dagger=R_{10}b_{z,0}^\dagger+R_{11}b_{-z,1}^\dagger.
\end{aligned}
\end{equation}
After passing through the beam splitter, the state prepared by Charlie in Eq.~\ref{eq:phi_SI} evolves as
\begin{equation}\begin{aligned}
&\ket{\xi}=\sum_{\vec{q}}\int d\vec{z}\xi(y,\vec{q},\vec{i},\vec{z},t)\prod_{j=1}^Nb_{z_j,q_j}^\dagger\ket{0}  ,  \\
&\xi(y,\vec{q},\vec{z},t)=\frac{\psi_l}{\sqrt{N!}}\sum_{\vec{i}}h_1(y,\vec{q},\vec{i},\vec{z},t)\prod_{j=1}^N U_{i_j 0}R_{i_j q_j}+\frac{\psi_r}{\sqrt{N!}}\sum_{\vec{i}}h_2(y,\vec{q},\vec{i},\vec{z},t)\prod_{j=1}^N U_{i_j 1}R_{i_j q_j},\\
&h_1(y,\vec{q},\vec{i},\vec{z},t)=\psi(Ny+\sum_{j=1}^N((-1)^{q_j}z_j-(-1)^{i_j}y)+\frac{NL}{2}-cNt)(\frac{\sigma^2}{\pi})^{N/4}\exp[-\frac{\sigma^2}{2}\sum_{j=1}^N((-1)^{q_j}z_j-(-1)^{i_j}y+y+\frac{L}{2}-ct)^2],\\
&h_2(y,\vec{q},\vec{i},\vec{z},t)=\psi(Ny-\sum_{j=1}^N((-1)^{q_j}z_j-(-1)^{i_j}y)-\frac{NL}{2}+cNt)(\frac{\sigma^2}{\pi})^{N/4}\exp[-\frac{\sigma^2}{2}\sum_{j=1}^N(-(-1)^{q_j}z_j+(-1)^{i_j}y+y-\frac{L}{2}+ct)^2].\\
\end{aligned}
\end{equation}
If Alice and Bob directly measure by projecting onto the basis 
\begin{equation}\label{SI_eq:measurement}
\ket{\vec{q},\vec{z}}=\prod_{j=1}^Nb_{z_j,q_j}^\dagger\ket{0} .
\end{equation}
The probability is given by
\begin{equation}\begin{aligned}
&P(\vec{q},\vec{z}|y)=|\bra{\vec{q},\vec{z}}\ket{\xi}|^2\\
&=|\psi_l|^2\sum_{\vec{i},\vec{i}'}h_1(y,\vec{q},\vec{i},\vec{z},t)h^*_1(y,\vec{q},\vec{i}',\vec{z},t)\prod_{j=1}^N (U_{i_j 0}R_{i_j q_j}U^*_{i'_j 0}R^*_{i'_j q_j})\\
&+|\psi_r|^2\sum_{\vec{i},\vec{i}'}h_2(y,\vec{q},\vec{i},\vec{z},t)h^*_2(y,\vec{q},\vec{i}',\vec{z},t)\prod_{j=1}^N (U_{i_j 1}R_{i_j q_j}U^*_{i'_j 1}R^*_{i'_j q_j})\\
&+\psi_l\psi_r^*\sum_{\vec{i},\vec{i}'}h_1(y,\vec{q},\vec{i},\vec{z},t)h^*_2(y,\vec{q},\vec{i}',\vec{z},t)\prod_{j=1}^N (U_{i_j 0}R_{i_j q_j}U^*_{i'_j 1}R^*_{i'_j q_j})\\
&+\psi_l^*\psi_r\sum_{\vec{i},\vec{i}'}h_2(y,\vec{q},\vec{i},\vec{z},t)h_1^*(y,\vec{q},\vec{i}',\vec{z},t)\prod_{j=1}^N (U_{i_j 1}R_{i_j q_j}U^*_{i'_j 0}R^*_{i'_j q_j}),\\
\end{aligned}
\end{equation}
\begin{equation}\begin{aligned}
\frac{\partial P(\vec{q},\vec{z}|y)}{\partial y}&=|\psi_l|^2\sum_{\vec{i},\vec{i}'}h_1(y,\vec{q},\vec{i},\vec{z},t)h^*_1(y,\vec{q},\vec{i}',\vec{z},t)(g_1(y,\vec{q},\vec{i},\vec{z},t)+g^*_1(y,\vec{q},\vec{i}',\vec{z},t))\prod_{j=1}^N (U_{i_j 0}R_{i_j q_j}U^*_{i'_j 0}R^*_{i'_j q_j})\\
&+|\psi_r|^2\sum_{\vec{i},\vec{i}'}h_2(y,\vec{q},\vec{i},\vec{z},t)h^*_2(y,\vec{q},\vec{i}',\vec{z},t)(g_2(y,\vec{q},\vec{i},\vec{z},t)+g^*_2(y,\vec{q},\vec{i}',\vec{z},t))\prod_{j=1}^N (U_{i_j 1}R_{i_j q_j}U^*_{i'_j 1}R^*_{i'_j q_j})\\
&+\psi_l\psi_r^*\sum_{\vec{i},\vec{i}'}h_1(y,\vec{q},\vec{i},\vec{z},t)h^*_2(y,\vec{q},\vec{i}',\vec{z},t)(g_1(y,\vec{q},\vec{i},\vec{z},t)+g^*_2(y,\vec{q},\vec{i}',\vec{z},t))\prod_{j=1}^N (U_{i_j 0}R_{i_j q_j}U^*_{i'_j 1}R^*_{i'_j q_j})\\
&+\psi_l^*\psi_r\sum_{\vec{i},\vec{i}'}h_2(y,\vec{q},\vec{i},\vec{z},t)h_1^*(y,\vec{q},\vec{i}',\vec{z},t)(g_2(y,\vec{q},\vec{i},\vec{z},t)+g^*_1(y,\vec{q},\vec{i}',\vec{z},t))\prod_{j=1}^N (U_{i_j 1}R_{i_j q_j}U^*_{i'_j 0}R^*_{i'_j q_j}),\\
\end{aligned}
\end{equation}
where we use the pulse shape defined in Eq.~\ref{eq:psi_k} and have defined
\begin{equation}\begin{aligned}
&\frac{\partial h_1(y,\vec{q},\vec{i},\vec{z},t)}{\partial y}=h_1(y,\vec{q},\vec{i},\vec{z},t) g_1(y,\vec{q},\vec{i},\vec{z},t),\quad \frac{\partial h_2(y,\vec{q},\vec{i},\vec{z},t)}{\partial y}=h_2(y,\vec{q},\vec{i},\vec{z},t) g_2(y,\vec{q},\vec{i},\vec{z},t),    \\
&g_1=-\beta^2(N-\sum_{j=1}^N(-1)^{i_j})(Ny+\sum_{j=1}^N((-1)^{q_j}z_j-(-1)^{i_j}y)+\frac{NL}{2}-cNt)\\
&\quad\quad\quad-\sigma^2\sum_{j=1}^N(1-(-1)^{i_j})((-1)^{q_j}z_j-(-1)^{i_j}y+y+\frac{L}{2}-ct),\\
&g_2=-\beta^2(N+\sum_{j=1}^N(-1)^{i_j})(Ny-\sum_{j=1}^N((-1)^{q_j}z_j-(-1)^{i_j}y)-\frac{NL}{2}+   cNt)\\
&\quad\quad\quad-\sigma^2\sum_{j=1}^N(1+(-1)^{i_j})(-(-1)^{q_j}z_j+(-1)^{i_j}y+y-\frac{L}{2}+ct),\\
\end{aligned}
\end{equation}
where the $\sigma^2$ terms in $g_{1,2}$ can be neglected since we take $\sigma\rightarrow0$. With the pulse shape defined in Eq.~\ref{eq:psi_k},  $h_{1,2}$ and $g_{1,2}$ are real-valued functions. For simplicity, we set $y = 0$ and calculate the FI in the vicinity of this point. Under this assumption, we find that $h_{1,2}$ is no longer dependent on the index vector $\vec{i}$, and thus we have

\begin{equation}\begin{aligned}
&\xi(0,\vec{q},\vec{z},t)=\frac{\psi_l}{\sqrt{N!}}h_1(0,\vec{q},\vec{z},t)\prod_{j=1}^N (\sum_{i_j=0,1}U_{i_j 0}R_{i_j q_j})+\frac{\psi_r}{\sqrt{N!}}h_2(0,\vec{q},\vec{z},t)\prod_{j=1}^N(\sum_{i_j=0,1} U_{i_j 1}R_{i_j q_j}).\\
\end{aligned}
\end{equation}
We further choose $R_{ij}=U_{ij}^*$, which gives
\begin{equation}\begin{aligned}
&\xi(0,\vec{q},\vec{z},t)=\frac{\psi_l}{\sqrt{N!}}h_1(0,\vec{q},\vec{z},t)\prod_{j=1}^N \delta_{q_j0}+\frac{\psi_r}{\sqrt{N!}}h_2(0,\vec{q},\vec{z},t)\prod_{j=1}^N\delta_{q_j1}.\\
\end{aligned}
\end{equation}
We can then identify the outcomes with nonvanishing probability as given by
\begin{equation}\label{eq:P_Charlie}
P(\vec{q}_0,\vec{z}|y)=|\psi_l|^2 h_1^2(y,\vec{q}_0,\vec{z},t),\quad P(\vec{q}_1,\vec{z}|y)=|\psi_r|^2 h_2^2(y,\vec{q}_1,\vec{z},t),
\end{equation}
where $\vec{q}_0=[0,0,\cdots,0]$, $\vec{q}_1=[1,1,\cdots,1]$.

In the case of $\vec{q} = \vec{q}_0$, we note that $h_1(0, \vec{q}_0, \vec{z}, t) = h_2(0, \vec{q}_0, \vec{z}, t)$, and we denote this common value as $h(0, \vec{q}_0, \vec{z}, t)$.
\begin{equation}\begin{aligned}
\frac{\partial P(\vec{q_0},\vec{z}|y)}{\partial y}
&=|\psi_l|^2h^2(0,\vec{q}_0,\vec{z},t)\sum_{\vec{i},\vec{i}'}(\sum_{j=1}^Nz_j+\frac{NL}{2}-cNt)2\beta^2(-Q-Q')\prod_{j=1}^N (U_{i_j 0}R_{i_j 0}U^*_{i'_j 0}R^*_{i'_j 0})\\
&+|\psi_r|^2h^2(0,\vec{q}_0,\vec{z},t)\sum_{\vec{i},\vec{i}'}(\sum_{j=1}^Nz_j+\frac{NL}{2}-cNt)2\beta^2(2N-Q-Q')\prod_{j=1}^N (U_{i_j 1}R_{i_j 0}U^*_{i'_j 1}R^*_{i'_j 0})\\
&+\psi_l\psi_r^*h^2(0,\vec{q}_0,\vec{z},t)\sum_{\vec{i},\vec{i}'}(\sum_{j=1}^Nz_j+\frac{NL}{2}-cNt)2\beta^2(-Q+N-Q')\prod_{j=1}^N (U_{i_j 0}R_{i_j 0}U^*_{i'_j 1}R^*_{i'_j 0})\\
&+\psi_l^*\psi_rh^2(0,\vec{q}_0,\vec{z},t)\sum_{\vec{i},\vec{i}'}(\sum_{j=1}^Nz_j+\frac{NL}{2}-cNt)2\beta^2(N-Q-Q')\prod_{j=1}^N (U_{i_j 1}R_{i_j 0}U^*_{i'_j 0}R^*_{i'_j 0}).\\
\end{aligned}
\end{equation}

In the case of $\vec{q} = \vec{q}_1$, we note that $h_1(0, \vec{q}_1, \vec{z}, t) = h_2(0, \vec{q}_1, \vec{z}, t)$, and we denote this common value as $h(0, \vec{q}_1, \vec{z}, t)$.
\begin{equation}\begin{aligned}
\frac{\partial P(\vec{q_1},\vec{z}|y)}{\partial y}&=|\psi_l|^2h^2(0,\vec{q}_1,\vec{z},t)\sum_{\vec{i},\vec{i}'}(-\sum_{j=1}^Nz_j+\frac{NL}{2}-cNt)2\beta^2(-Q-Q')\prod_{j=1}^N (U_{i_j 0}R_{i_j 1}U^*_{i'_j 0}R^*_{i'_j 1})\\
&+|\psi_r|^2h^2(0,\vec{q}_1,\vec{z},t)\sum_{\vec{i},\vec{i}'}(-\sum_{j=1}^Nz_j+\frac{NL}{2}-cNt)2\beta^2(2N-Q-Q')\prod_{j=1}^N (U_{i_j 1}R_{i_j 1}U^*_{i'_j 1}R^*_{i'_j 1})\\
&+\psi_l\psi_r^*h^2(0,\vec{q}_1,\vec{z},t)\sum_{\vec{i},\vec{i}'}(-\sum_{j=1}^Nz_j+\frac{NL}{2}-cNt)2\beta^2(-Q+N-Q')\prod_{j=1}^N (U_{i_j 0}R_{i_j 1}U^*_{i'_j 1}R^*_{i'_j 1})\\
&+\psi_l^*\psi_rh^2(0,\vec{q}_1,\vec{z},t)\sum_{\vec{i},\vec{i}'}(-\sum_{j=1}^Nz_j+\frac{NL}{2}-cNt)2\beta^2(N-Q-Q')\prod_{j=1}^N (U_{i_j 1}R_{i_j 1}U^*_{i'_j 0}R^*_{i'_j 1}).\\
\end{aligned}
\end{equation}

To further simplify the discussion, we assume 
\begin{equation}\label{eq:R_U_equal}
\begin{aligned}
R=U^*=\frac{1}{\sqrt{2}}\left[\begin{matrix}
1 & 1\\
-1 & 1
\end{matrix}\right],
\end{aligned}
\end{equation}

\begin{equation}\begin{aligned}
&\frac{\partial P(\vec{q_0},\vec{z}|y)}{\partial y}=\frac{1}{2^{2N}}(\sum_{j=1}^Nz_j+\frac{NL}{2}-cNt)2\beta^2h^2(0,\vec{q}_0,\vec{z},t)\sum_{Q,Q'}C_N^Q C_N^{Q'} \\
&\times[|\psi_l|^2(-Q-Q')+|\psi_r|^2(2N-Q-Q')(-1)^{Q+Q'}+\psi_l\psi_r^*(-Q+N-Q')(-1)^{Q'}+\psi_l^*\psi_r(N-Q-Q')(-1)^{Q}],\\
\end{aligned}
\end{equation}
\begin{equation}\begin{aligned}
&\frac{\partial P(\vec{q_1},\vec{z}|y)}{\partial y}=\frac{1}{2^{2N}}(-\sum_{j=1}^Nz_j+\frac{NL}{2}-cNt)2\beta^2h^2(0,\vec{q}_1,\vec{z},t)\sum_{Q,Q'}C_N^Q C_N^{Q'} \\
&\times[|\psi_l|^2(-Q-Q')(-1)^{Q+Q'}+|\psi_r|^2(2N-Q-Q')+\psi_l\psi_r^*(-Q+N-Q')(-1)^{Q}+\psi_l^*\psi_r(N-Q-Q')(-1)^{Q'}].\\
\end{aligned}
\end{equation}
We can then calculate the Fisher information in the vicinity of  $y=0$ and with operations in Eq. \ref{eq:R_U_equal}, 
\begin{equation}
F=\int d\vec{z}(\frac{\partial P(\vec{q_0},\vec{z}|y)}{\partial y})^2\frac{1}{P(\vec{q_0},\vec{z}|y)}+\int d\vec{z}(\frac{\partial P(\vec{q_1},\vec{z}|y)}{\partial y})^2\frac{1}{P(\vec{q_1},\vec{z}|y)},
\end{equation}
\begin{equation}\begin{aligned}
&\int d\vec{z}(\frac{\partial P(\vec{q_0},\vec{z}|y)}{\partial y})^2\frac{1}{P(\vec{q_0},\vec{z}|y)}\\
&=\frac{1}{|\psi_l|^2}\frac{1}{2^{4N}}4\beta^4 \int d\vec{z} (\sum_{j=1}^Nz_j+\frac{NL}{2}-cNt)^2h^2(0,\vec{q}_0,\vec{z},t)\sum_{Q_1,Q_1',Q_2,Q_2'}C_N^{Q_1} C_N^{Q_1'} C_N^{Q_2} C_N^{Q_2'} \\
&\times[|\psi_l|^2(-Q_1-Q_1')+|\psi_r|^2(2N-Q_1-Q_1')(-1)^{Q_1+Q_1'}+\psi_l\psi_r^*(-Q_1+N-Q_1')(-1)^{Q_1'}+\psi_l^*\psi_r(N-Q_1-Q_1')(-1)^{Q_1}]\\    
&\times[|\psi_l|^2(-Q_2-Q_2')+|\psi_r|^2(2N-Q_2-Q_2')(-1)^{Q_2+Q_2'}+\psi_l\psi_r^*(-Q_2+N-Q_2')(-1)^{Q_2'}+\psi_l^*\psi_r(N-Q_2-Q_2')(-1)^{Q_2}],\\   
\end{aligned}
\end{equation}
\begin{equation}\begin{aligned}
&\int d\vec{z}(\frac{\partial P(\vec{q_1},\vec{z}|y)}{\partial y})^2\frac{1}{P(\vec{q_1},\vec{z}|y)}\\
&=\frac{1}{|\psi_r|^2}\frac{1}{2^{4N}}4\beta^4 \int d\vec{z} (-\sum_{j=1}^Nz_j+\frac{NL}{2}-cNt)^2h^2(0,\vec{q}_1,\vec{z},t)\sum_{Q_1,Q_1',Q_2,Q_2'}C_N^{Q_1} C_N^{Q_1'} C_N^{Q_2} C_N^{Q_2'} \\
&\times[|\psi_l|^2(-Q_1-Q_1')(-1)^{Q_1+Q_1'}+|\psi_r|^2(2N-Q_1-Q_1')+\psi_l\psi_r^*(-Q_1+N-Q_1')(-1)^{Q_1}+\psi_l^*\psi_r(N-Q_1-Q_1')(-1)^{Q_1'}]\\    
&\times[|\psi_l|^2(-Q_2-Q_2')(-1)^{Q_2+Q_2'}+|\psi_r|^2(2N-Q_2-Q_2')+\psi_l\psi_r^*(-Q_2+N-Q_2')(-1)^{Q_2}+\psi_l^*\psi_r(N-Q_2-Q_2')(-1)^{Q_2'}].\\   
\end{aligned}
\end{equation}
Note that
\begin{equation}
\int d\vec{z} (\sum_{j=1}^Nz_j+\frac{NL}{2}-cNt)^2h^2(0,\vec{q}_0,\vec{z},t)=\int d\vec{z} (-\sum_{j=1}^Nz_j+\frac{NL}{2}-cNt)^2h^2(0,\vec{q}_1,\vec{z},t)=1/2\beta^2,
\end{equation}
\begin{equation}\begin{aligned}
&\sum_{Q,Q'}C_N^Q C_N^{Q'} [|\psi_l|^2(-Q-Q')+|\psi_r|^2(2N-Q-Q')(-1)^{Q+Q'}+\psi_l\psi_r^*(-Q+N-Q')(-1)^{Q'}+\psi_l^*\psi_r(N-Q-Q')(-1)^{Q}]\\
&=-|\psi_l|^22^{2N}N,
\end{aligned}
\end{equation}
\begin{equation}\begin{aligned}
&\sum_{Q,Q'}C_N^Q C_N^{Q'} [|\psi_l|^2(-Q-Q')(-1)^{Q+Q'}+|\psi_r|^2(2N-Q-Q')+\psi_l\psi_r^*(-Q+N-Q')(-1)^{Q}+\psi_l^*\psi_r(N-Q-Q')(-1)^{Q'}]\\
&=|\psi_r|^22^{2N}N.
\end{aligned}
\end{equation}
Combining the above calculation, we can find the FI of estimating $y$
\begin{equation}
F=2\beta^2N^2|\psi_l|^2+2\beta^2N^2|\psi_r|^2=2\beta^2N^2,
\end{equation}
which is slightly worse than the predicted QFI, but still achieves the Heisenberg limit over $N$.

\subsection{Error probability in detecting potential cheating}

We now analyze the security of this approach. We still assume the honest prover Charlie is located at position $y = 0$, while the cheaters attempt to impersonate him at a fake position $y' \neq y$. Alice and Bob perform the measurement described in Eq. \ref{SI_eq:measurement} and obtain an estimated position $y''$. For generality, we again allow $y''$ to differ from both $y$ and $y'$.
Alice and Bob adopt the following two-step decision strategy:

Step (1): They compute the total probability associated with the observed measurement outcomes
\begin{equation}
\begin{aligned}
P_{\text{tot}}=\sum_{\vec{q}}\int d\vec{z}P(\vec{q},\vec{z}).
\end{aligned}
\end{equation}
If Alice and Bob find that $P_{\text{tot}} < 1$—that is, if they observe that in some instances not all $N$ photons from the original state $\ket{\psi}$ are detected—they conclude that cheaters are present. Noting that $\delta_{p,q} = \int d\vec{z}, h_p(y, \vec{q}, \vec{i}, \vec{z}, t), h_q^*(y, \vec{q}, \vec{i}', \vec{z}, t) \leq 1$, and that this quantity becomes independent of $\vec{q}$ after integration, in the scenario where the cheaters act directly on the original states sent by Alice and Bob, we have
\begin{equation}\begin{aligned}
P_{\text{tot}}&= \sum_{\vec{i},\vec{i}'}\bigg[|\psi_l|^2\delta_{11}\prod_{j=1}^N \sum_{q_j}(U'_{i_j 0}R_{i_j q_j}U'^*_{i'_j 0}R^*_{i'_j q_j})+|\psi_r|^2\delta_{22}\prod_{j=1}^N \sum_{q_j}(U'_{i_j 1}R_{i_j q_j}U'^*_{i'_j 1}R^*_{i'_j q_j})\\
&+\psi_l\psi_r^*\delta_{12}\prod_{j=1}^N\sum_{q_j} (U'_{i_j 0}R_{i_j q_j}U'^*_{i'_j 1}R^*_{i'_j q_j})+\psi_l^*\psi_r\delta_{21}\prod_{j=1}^N \sum_{q_j}(U'_{i_j 1}R_{i_j q_j}U'^*_{i'_j 0}R^*_{i'_j q_j})\bigg]\\ 
&=\sum_{\vec{i}}\bigg[|\psi_l|^2\delta_{11}\prod_{j=1}^N (U'_{i_j 0}U'^*_{i_j 0})+|\psi_r|^2\delta_{22}\prod_{j=1}^N (U'_{i_j 1}U'^*_{i_j 1})+\psi_l\psi_r^*\delta_{12}\prod_{j=1}^N(U'_{i_j 0}U'^*_{i_j 1})+\psi_l^*\psi_r\delta_{21}\prod_{j=1}^N (U'_{i_j 1}U'^*_{i_j 0})\bigg].\\
\end{aligned}\end{equation}
We then still use an argument similar to Eq.~\ref{SI_eq:U_U_delta}
\begin{equation}\begin{aligned}
&\sum_{\vec{i}}\prod_{j=1}^NU_{i_j,p}'^*U_{i_j,q}'\delta_{p+1,q+1}=\sum_{Q=0}^NC_N^Qa^{N-Q}b^Q\delta_{p+1,q+1} \\
&\leq\sum_{Q=0}^NC_N^Q|a|^{N-Q}|b|^Q\delta_{p+1,q+1}\leq \sum_{Q=0}^NC_N^Q|a|^{N-Q}|b|^Q=(|a|+|b|)^N\leq|u_p'|^N|u_q'|^N,\\
\end{aligned}
\end{equation}
\begin{equation}\begin{aligned}\label{SI_eq:ab_u}
&(|a|+|b|)^2=(|U'_{0p}||U'_{0q}|+|U'_{1p}||U'_{1q}|)^2\leq(|U'_{0p}|^2+|U'_{1p}|^2)(|U'_{0q}|^2+|U'_{1q}|^2)=|u_p'|^2|u_q'|^2,
\end{aligned}
\end{equation}
where $a=U_{0p}'U_{0q}'^*$, $b=U_{1p}'U_{1q}'^*$. Since we cannot have $|u_0'| = |u_1'| = 1$ except in the special case considered below in Eq.~\ref{SI_eq:special_U'}, the  probability is bounded similar to Eq. \ref{SI_eq:product_phi_gamma} (excluding these special cases of $U'$ discussed below).
\begin{equation}
P_{\text{tot}}\leq \max\{|\psi_l|^2,|\psi_r|^2\}.
\end{equation}
Noting that
$U_{00}' = V_{00} P_{00}$, $U_{10}' = V_{01} Q_{10}$, $U_{01}' = W_{01} P_{10}$, and $U_{11}' = W_{00} Q_{00}$,
the condition $|u_0'| = |u_1'| = 1$ can be satisfied only in the special cases where
\begin{equation}\label{SI_eq:special_U'}
U'=\left[\begin{matrix}
1 & 0\\
0 & e^{i\theta}
\end{matrix}\right]\quad \text{or}\quad \left[\begin{matrix}
0 & e^{i\alpha}\\
e^{i\beta} & 0
\end{matrix}\right],
\end{equation}
which correspond to the cheaters either directly reflecting the state or applying only a phase shift.
Such special cases can be excluded—with some overhead—by performing a projection onto the state $\prod_{j=1}^N a_{z_j, q_j}^\dagger \ket{0}$ without using the beam splitter. In this configuration, all $N$ photons will be detected on one side—either Alice or Bob—with probabilities $|\psi_l|^2$ and $|\psi_r|^2$, respectively. These cases can be distinguished, as Charlie’s application of the unitary $U$ will cause the $N$ photons to be detected on both sides.

Step (2): If Alice and Bob always detect all $N$ photons and have already excluded the special cases described in Eq.~\ref{SI_eq:special_U'}, they then proceed with the following strategy to rule out the possibility that the cheaters have discarded the original states and prepared their own instead.
Let $Q_1(\vec{q}, \vec{z})$ and $Q_2(\vec{q}, \vec{z})$ denote the honest probability distributions corresponding to two different random states, and let $Q_0(\vec{q}, \vec{z})$ denote the (unknown) distribution of the cheaters’ prepared state. Alice and Bob can potentially detect the presence of cheaters by checking whether the observed outcome probabilities vary under different random input choices.
Suppose there are $M$ total measurement outcomes, with $M/2$ samples corresponding to $Q_1$ and $M/2$ to $Q_2$. Since the form of $Q_0(\vec{q}, \vec{z})$ is unknown, Alice and Bob cannot directly compare their data to it. Instead, they estimate the observed distribution empirically from the samples. After collecting $M$ samples, they construct estimated distributions $\hat{P}_{1,2}(x)$ using the procedure described in Lemma~\ref{proposition:Q_hat}, where $x$ denotes the measurement outcome, used as a simplified label for the pair $(\vec{q}, \vec{z})$.
They then compute the total variation distance $TV(Q_1 \otimes Q_2, \hat{P}_1 \otimes \hat{P}_2)$. If this distance exceeds a fixed threshold $\epsilon$, they conclude that cheating has occurred. Conversely, if $TV(Q_1 \otimes Q_2, \hat{P}_1 \otimes \hat{P}_2) < \epsilon$, they conclude that no cheating is detected.

We will first prove the following lemmas regarding the deviation of the estimated distribution from the expected distributions with finite samples.
\begin{lemma}\label{proposition:Q_hat}
For a probability distribution $Q(x)$ over a continuous variable $x$, with $\int dx, Q(x) = 1$, suppose we obtain $M$ independent samples ${X_1, X_2, \dots, X_M}$ drawn from $Q(x)$. We then construct the estimated distribution $\hat{Q}(x) = \frac{1}{M} \sum_{i=1}^M \frac{1}{h} K\left( \frac{x - X_i}{h} \right)$, where $K(x) = \exp(-x^2/2)/\sqrt{2\pi}$. For any $\varepsilon > 0$, the probability that the estimated distribution $\hat{Q}$ deviates from the true distribution $Q$, as quantified by the total variation distance $TV(Q, \hat{Q})$, is given by
\begin{equation}
P(TV(Q,\hat{Q})>\varepsilon)\leq \exp(-M\xi).
\end{equation}
\end{lemma}
\begin{proof}
We first notice that $\hat{Q}$ is a function depending on the samples where each $X_i$ follows the distribution $Q(x)$, so we take the expectation for this function and upper bound as
\begin{equation}\label{SI_eq:TV_Q_Q}
TV(Q,\hat{Q})\leq TV(Q,\mathbb{E}[\hat{Q}])+TV(\mathbb{E}[\hat{Q}],\hat{Q}).
\end{equation}
We now evaluate
\begin{equation}
\begin{aligned}
\mathbb{E}[\hat{Q}]&=\frac{1}{Mh}\sum_{i=1}^M\int dX_i K(\frac{x-X_i}{h})Q(X_i)=\frac{1}{M}\sum_{i=1}^M\int du K(u)Q(x-uh)\\
&=\frac{1}{M}\sum_{i=1}^M\int du K(u)[Q(x)-hu Q'(x)+\frac{h^2u^2}{2}Q''(x)+o(h^2u^2)]\\
&=Q(x)+\frac{h^2}{2}\mu Q''(x)+o(h^2),
\end{aligned}
\end{equation}
where we use $u=(x-X_i)/h$ in the second equality and define $\mu=\int du K(u)u^2$. And we have
\begin{equation}
\begin{aligned}\label{SI_eq:TV_Q_EQ}
TV(Q,\mathbb{E}[\hat{Q}])=\frac{1}{2}\int dx|Q(x)-\mathbb{E}[\hat{Q}](x)|=\frac{h^2}{4}\mu \int dx|Q''(x)|+o(h^2).
\end{aligned}
\end{equation}

For $TV(\mathbb{E}[\hat{Q}],\hat{Q})$, since it depends on the actual outcome $X_i$, we will bound the probability $P(TV(\mathbb{E}[\hat{Q}],\hat{Q})>\varepsilon)$,  we will first bound the expectation $\mathbb{E}[TV(\mathbb{E}[\hat{Q}],\hat{Q})]$ and then calculate the deviation from this expectation. 
\begin{equation}
\mathbb{E}[TV(\mathbb{E}[\hat{Q}],\hat{Q})]=\frac{1}{2}\int dx \mathbb{E}\bigg[\bigg|\mathbb{E}[\hat{Q}](x)-\hat{Q}(x)\bigg|\bigg]\leq \frac{1}{2}\int dx \sqrt{\text{Var}[\hat{Q}(x)]},
\end{equation}
where we use the Cauchy–Schwarz inequality.
\begin{equation}\begin{aligned}
&\mathbb{E}[\frac{1}{h}K(\frac{x-X_i}{h})]=\int \frac{1}{h}K(\frac{x-X_i}{h})Q(X_i)dX_i=\int du K(u)Q(x-hu)\\
&=\int K(u)[Q(x)-hu Q'(x)+\frac{h^2u^2}{2}Q''(x)+o(h^2)]=Q(x)+\frac{h^2}{2}\mu Q''(x)+o(h^2) ,   
\end{aligned}
\end{equation}
\begin{equation}\begin{aligned}
&\mathbb{E}[\frac{1}{h^2}K^2(\frac{x-X_i}{h})]=\int \frac{1}{h^2}K^2(\frac{x-X_i}{h})Q(X_i)dX_i=\int du\frac{1}{h} K^2(u)Q(x-hu)\\
&=\int \frac{1}{h}K^2(u)[Q(x)-hu Q'(x)+\frac{h^2u^2}{2}Q''(x)+o(h^2)]=Q(x)\frac{R}{h}+O(h) ,
\end{aligned}
\end{equation}
\begin{equation}
\text{Var}[\hat{Q}(x)]=\frac{1}{M}\text{Var}[\frac{1}{h}K(\frac{x-X_i}{h})]=\frac{1}{M}(Q(x)\frac{R}{h}-Q^2(x)+O(h))\leq Q(x)\frac{R}{Mh},
\end{equation}
where we define $R=\int K^2(u)du$. And hence we have
\begin{equation}\label{SI_eq:E_TV}
\mathbb{E}[TV(\mathbb{E}[\hat{Q}],\hat{Q})]\leq \frac{1}{2}\sqrt{\frac{R}{Mh}}\int dx \sqrt{Q(x)}.
\end{equation}
Let us now bound the deviation of $TV(\mathbb{E}[\hat{Q}],\hat{Q})$ from its expectation, this can be derived from McDiarmid's inequality, which states that 
\begin{equation}\begin{aligned}
& P\bigg(f(X_1,X_2,\cdots,X_i,\cdots,X_M)-\mathbb{E}[f(X_1,X_2,\cdots,X_i',\cdots,X_M)]>t\bigg)\leq \exp(-2t^2/\sum_{i=1}^Mc_i^2),\\
& \sup_{X_i'}|f(X_1,X_2,\cdots,X_i,\cdots,X_M)-f(X_1,X_2,\cdots,X_i',\cdots,X_M)|\leq c_i   . 
\end{aligned}
\end{equation}
For the case of $f(X_1,X_2,\cdots,X_M)=TV(\mathbb{E}[\hat{Q}],\hat{Q})$
\begin{equation}\begin{aligned}
& \sup_{X_i'}|f(X_1,X_2,\cdots,X_i,\cdots,X_M)-f(X_1,X_2,\cdots,X_i',\cdots,X_M)|\\
&=\frac{1}{2}\sup_{X_i'}\left|\int dx\left|\mathbb{E}[\hat{Q}](x)-\hat{Q}(x|X_1,\cdots,X_i,\cdots,X_M)\right|-\int dx\left|\mathbb{E}[\hat{Q}](x)-\hat{Q}(x|X_1,\cdots,X_i',\cdots,X_M)\right|\right|\\
&\leq \frac{1}{2}\sup_{X_i'}\int dx\left|\hat{Q}(x|X_1,\cdots,X_i,\cdots,X_M)-\hat{Q}(x|X_1,\cdots,X_i',\cdots,X_M)\right|\\
&=\frac{1}{2Mh}\sup_{X_i'}\int dx\left|K(\frac{x-X_i}{h})-K(\frac{x-X_i'}{h})\right|\leq \frac{1}{2Mh}\sup_{X_i'}\int dx\left(K(\frac{x-X_i}{h})+K(\frac{x-X_i'}{h})\right)\\
&=\frac{1}{M}.
\end{aligned}
\end{equation}
We have thus proved
\begin{equation}
P(TV(\mathbb{E}[\hat{Q}],\hat{Q})-\mathbb{E}[TV(\mathbb{E}[\hat{Q}],\hat{Q})]>t)\leq \exp(-2Mt^2).
\end{equation}
Combined with Eq.~\ref{SI_eq:E_TV}, we have
\begin{equation}
P\left(TV(\mathbb{E}[\hat{Q}],\hat{Q})>t+\frac{1}{2}\sqrt{\frac{R}{Mh}}\int dx \sqrt{Q(x)}\right)\leq \exp(-2Mt^2).
\end{equation}
We further combine Eq. \ref{SI_eq:TV_Q_EQ} and Eq. \ref{SI_eq:TV_Q_Q} and get
\begin{equation}
P\left(TV(\mathbb{E}[\hat{Q}],\hat{Q})>t+\frac{1}{2}\sqrt{\frac{R}{Mh}}\int dx \sqrt{Q(x)}+\frac{h^2}{4}\mu \int dx|Q''(x)|\right)\leq \exp(-2Mt^2).
\end{equation}
We choose $h=M^{-1/5}$, which then gives
\begin{equation}
P\left(TV(\mathbb{E}[\hat{Q}],\hat{Q})>t+\frac{1}{2M^{2/5}}\sqrt{R}\int dx \sqrt{Q(x)}+\frac{1}{4M^{2/5}}\mu \int dx|Q''(x)|\right)\leq \exp(-2Mt^2).
\end{equation}
So, for sufficiently large $M$, and any $\varepsilon$, we can choose $t=\varepsilon-\frac{1}{2M^{2/5}}\sqrt{R}\int dx \sqrt{Q(x)}-\frac{1}{4M^{2/5}}\mu \int dx|Q''(x)|$, $\xi=2t^2$, which completes the proof.

\end{proof}

%which means Alice and Bob compare the probability distribution of $M$ samples for the case of there is no cheater and hence the probability distribution is $Q_1^{\otimes {M/2}} \otimes Q_2^{\otimes {M/2}}$ and the case the probability is fixed as $Q_0^{\otimes M}$. Since they do not know which is $Q_0$ is actually the probability when cheaters are present, they go through all the possible $Q_0$ and decide the cheaters exist whenever $Q_0^{\otimes M}$ is probable than $Q_1^{\otimes {M/2}} \otimes Q_2^{\otimes {M/2}}$ given the outcomes $\vec{x}$. By assuming the $Q_0$ is the optimal approximation to $Q_1^{\otimes {M/2}} \otimes Q_2^{\otimes {M/2}}$. 
\begin{lemma}
Given the probability distribution $A,B,C$, the total variation distance has the following properties
\begin{equation}
\begin{aligned}
&TV(A,C)+TV(B,D)\geq  TV(A\otimes B, C\otimes D),\\
&TV(A\otimes B, C^{\otimes 2})\geq\frac{1}{2}TV(A,B).\\
\end{aligned}
\end{equation}
\end{lemma}
\begin{proof}
Combining the following two properties, we can easily prove the lower bound
\begin{equation}
\begin{aligned}
&TV(A\otimes B, C^{\otimes 2})\geq\max\{TV(A, C),TV( B, C)\},\\
&TV(A,B)\leq TV(A,C)+TV(B,C)\leq 2\max\{TV(A, C),TV(B, C)\}.
\end{aligned}
\end{equation}
For the upper bound
\begin{equation}
\begin{aligned}
&TV(A\otimes B,C\otimes D)=\frac{1}{2}\int dxdy|A(x)B(y)-C(x)D(y)|\\
&\leq \frac{1}{2}\int dx dy[|A(x)-C(x)|B(y)+|B(y)-D(y)|C(x)]=TV(A,C)+TV(B,D).
\end{aligned}
\end{equation}

\end{proof}

With the above lemmas in place, we are now ready to bound the error probability. There are three possible scenarios to consider, depending on whether cheaters are present and which strategies they employ:

Scenarios 1: When no cheaters are present, Alice and Bob will always pass step (1), as they detect all $N$ photons, and will proceed to make their decision based on step (2). In this case, the estimated probability distributions satisfy $\hat{P}_1 = \hat{Q}_1$ and $\hat{P}_2 = \hat{Q}_2$.
\begin{equation}
\begin{aligned}
&TV(Q_1\otimes Q_2, \hat{Q}_1\otimes\hat{Q}_2)\leq  TV(Q_1,\hat{Q}_1)+TV(Q_2,\hat{Q}_2).
\end{aligned}
\end{equation}
According to Lemma \ref{proposition:Q_hat}, choosing $\varepsilon_2=\frac{1}{2}\epsilon$, we have
\begin{equation}\begin{aligned}
&P(TV(Q_1\otimes Q_2, \hat{Q}_1\otimes\hat{Q}_2)>\epsilon)\leq P(TV(Q_1,\hat{Q}_1)+TV(Q_2,\hat{Q}_2)\geq 2\varepsilon_2)\\
&\leq1-(1- P(TV(Q_1,\hat{Q}_1)>\varepsilon_2))(1- P(TV(Q_2,\hat{Q}_2)>\varepsilon_2))\leq 2\exp(-M\xi_2) ,
\end{aligned}
\end{equation}
where $\xi_2=\min\{\xi_2(Q_1),\xi_2(Q_2)\}$, $\xi_2(Q_{1,2})=\frac{\epsilon}{2}-\frac{1}{2M^{2/5}}\sqrt{R}\int dx \sqrt{Q_{1,2}(x)}-\frac{1}{4M^{2/5}}\mu \int dx|Q_{1,2}''(x)|$.
And hence we find the error probability 
\begin{equation}
P(\text{cheaters}|\text{honest})\leq2 \exp(-M\xi_2).
\end{equation}

Scenarios 2: When cheaters are present and apply operations directly to the original states sent by Alice and Bob, the probability of detecting all $N$ photons in a single sample is bounded by $P_{\text{tot}} \leq \max\{|\psi_l|^2, |\psi_r|^2\}$. Therefore, the probability of not detecting the presence of cheaters after $M$ samples is bounded by
\begin{equation}
P_1(\text{honest}|\text{cheaters})\leq \max\{|\psi_l|^{2M},|\psi_r|^{2M}\}.
\end{equation}
Note that even if Alice and Bob fail to detect the cheaters in step (1), they may still succeed in identifying them in step (2). Thus, the actual error probability could be even lower; however, we use the above bound as a conservative estimate.

Scenarios 3: When cheaters are present and they discard the original states from Alice and Bob to prepare their own, the forged state must still pass the test in step (1). The probability of failing to detect the cheaters in this scenario is bounded as follows. Since the cheaters’ prepared state is independent of Alice and Bob’s random choices, the estimated distributions will be close to a fixed distribution $Q_0$. In particular, we expect $\hat{P}_1 = \hat{Q}_0$ and $\hat{P}_2 = \hat{Q}_0'$. In this case
\begin{equation}
\begin{aligned}
&TV(Q_1\otimes Q_2, \hat{Q}_0\otimes\hat{Q}'_0)\geq TV(Q_1\otimes Q_2, {Q}_0^{\otimes 2})-TV({Q}_0^{\otimes 2}, \hat{Q}_0\otimes\hat{Q}_0')\\
&\geq \frac{1}{2}TV(Q_1,Q_2)-TV(Q_0,\hat{Q}_0)-TV(Q_0,\hat{Q}_0').
\end{aligned}
\end{equation}
The total variation distance $TV(Q_1, Q_2)$ is a constant between 0 and 1. For example, in the case where $\psi_l = \frac{\sqrt{3}}{2}$ and $\psi_r = \frac{1}{2}$ (and vice versa for the other state), with $y'' = 0$ and $R, U$ chosen as in Eq.~\ref{eq:R_U_equal}, we find $TV(Q_1, Q_2) = \frac{1}{2}$. According to Lemma \ref{proposition:Q_hat}, as long as $\epsilon < \frac{1}{2} TV(Q_1, Q_2)$, we can set $\varepsilon_1 = \frac{1}{4} TV(Q_1, Q_2) - \frac{1}{2} \epsilon$, so that
\begin{equation}
\begin{aligned}
&P(TV(Q_1\otimes Q_2, \hat{Q}_0\otimes\hat{Q}'_0)<\epsilon)\leq P(TV(Q_0,\hat{Q}_0)+TV(Q_0,\hat{Q}_0')\geq 2\varepsilon_1)\\
&\leq 1-(1-P(TV(Q_0,\hat{Q}_0)>\varepsilon_1))^2\leq 2\exp(-M\xi_1),
\end{aligned}
\end{equation}
where $\xi_1=\frac{1}{2}TV(Q_1,Q_2)-\frac{\epsilon}{2}-\frac{1}{2M^{2/5}}\sqrt{R}\int dx \sqrt{Q_0(x)}-\frac{1}{4M^{2/5}}\mu \int dx|Q_0''(x)|$.
And hence we find that, when the cheater exist, the error probability that Alice and Bob conclude the cheaters do not exist is
\begin{equation}
P_2(\text{honest}|\text{cheaters})\leq 2\exp(-M\xi_1).
\end{equation}

So, the total error probability 
\begin{equation}\begin{aligned}
&P\leq P(\text{cheaters}|\text{honest})+\max\{P_1(\text{honest}|\text{cheaters}),P_2(\text{honest}|\text{cheaters})\}\\
&\leq 2\exp(-M\xi_2)+\max\{|\psi_l|^{2M},|\psi_r|^{2M},2\exp(-M\xi_1)\}.
\end{aligned}
\end{equation}

\end{document}